\documentclass[letterpaper, 10 pt, conference]{ieeeconf}  
\IEEEoverridecommandlockouts  
\usepackage[utf8]{inputenc}

\usepackage{xcolor}
\usepackage{graphicx}
\usepackage{graphics}
\usepackage{subcaption}
\usepackage{amsmath}
\usepackage[version=4]{mhchem}
\usepackage{siunitx}
\usepackage{longtable,tabularx}
\usepackage{algorithm}
\usepackage{algorithmicx}
\usepackage{algpseudocode}
\usepackage{amsmath}
\usepackage{amssymb}
\usepackage{proof}
\usepackage{epsfig}
\usepackage{calc}

\usepackage{amsthm}
\usepackage{accents}
\usepackage{cite}

\newtheorem{remark}{Remark}
\newtheorem{lemma}{Lemma}
\newtheorem{assumption}{Assumption}
\newtheorem{theorem}{Theorem}

\usepackage{xcolor}

\title{\LARGE \bf
Koopman Operator Based Modeling and Control of Rigid Body Motion Represented by Dual Quaternions
}

\author{Vrushabh Zinage \and Efstathios Bakolas
\thanks{This research has been supported  in part by NSF award CMMI-1937957.}
\thanks{Vrushabh Zinage (graduate student) and Efstathios Bakolas (Associate Professor) are with the Department of Aerospace Engineering and Engineering Mechanics,
The University of Texas at Austin, Austin, Texas 78712-1221, USA, {Emails:\tt\small vrushabh.zinage@utexas.edu; bakolas@austin.utexas.edu}}
}

\begin{document}

\bibliographystyle{IEEEtran} 

\maketitle
\thispagestyle{empty}
\pagestyle{empty}

\begin{abstract}

In this paper, we systematically derive a finite set of Koopman based observables to construct a lifted linear state space model that describes the rigid body dynamics based on the dual quaternion representation. In general, the Koopman operator is a linear infinite dimensional operator, which means that the derived linear state space model of the rigid body dynamics will be infinite-dimensional, which is not suitable for modeling and control design purposes. Recently, finite approximations of the operator computed by means of methods like the Extended Dynamic Mode Decomposition (EDMD) have shown promising results for different classes of problems. However, without using an appropriate set of observables in the EDMD approach, there can be no guarantees that the computed approximation of the nonlinear dynamics is sufficiently accurate. The major challenge in using the Koopman operator for constructing a linear state space model is the choice of observables. State-of-the-art methods in the field compute the approximations of the observables by using neural networks, standard radial basis functions (RBFs), polynomials or heuristic approximations of these functions. However, these observables might not provide a sufficiently accurate approximation or representation of the dynamics. In contrast, we first show the pointwise convergence of the derived observable functions to zero, thereby allowing us to choose a finite set of these observables. Next, we use the derived observables in EDMD to compute the lifted linear state and input matrices for the rigid body dynamics. Finally, we show that an LQR type (linear) controller, which is designed based on the truncated linear state space model, can steer the rigid body to a desired state while its performance is commensurate with that of a nonlinear controller. The efficacy of our approach is demonstrated through numerical simulations.
\end{abstract}

\section{Introduction\label{sec:introduction}}
We consider the problem of modeling and control of the dual quaternion based representation of rigid body motion using the Koopman operator framework. In particular, we propose a systematic way to describe the rigid body dynamics in terms of a linear system which is defined over a \textit{lifted} statee space spanned by the so-called Koopman based observables. The main advantage of utilizing the Koopman operator is that it explicitly accounts for nonlinearities in the dynamics unlike methods which rely on (local) linearization of the (nonlinear) dynamics about a point. The price that one has to pay when using the Koopman operator framework is that the lifted state space is in general infinite-dimensional and thus any meaningful finite-dimensional approximation (truncation) of the lifted state space will have higher dimension than the original nonlinear system model. The states of the lifted (linear) model are (nonlinear) functions of the states of the original (nonlinear) system which are known as observables or basis functions. Finding these observables is in general a complex task given that they do not exist systematic methods for their characterization for general nonlinear systems. In this paper, we derive in a systematic way a set of observables for the rigid body motion described in terms of dual quaternions and we subsequently propose simple linear control design techniques based on the lifted linear system associated with the latter observable. It turns out that these linear controllers perform similarly with a benchmark nonlinear controller for this particular system. 

\textit{Literature review:} The motion of a rigid body (both its position and attitude) can be represented in a compact and efficient manner through dual quaternions. This representation, takes automatically into account the natural coupling between the rotation and translation of a rigid body which thereby allows us to design a single controller which can control both the attitude and position of the rigid body simultaneously. Dual quaternions have been successfully applied to rigid body control \cite{pham2010position_rigid_body_1,perez2002dual_rigid_body_2,perez2002dual_rigid_body_3,wang2010feedback_rigid_body_4,han2008control_rigid_body_5,schlanbusch2012stability_dual_quaternion_automatica_2}, manipulator robots \cite{da2021robust_dual_quaternions_automatica}, inverse kinematic analysis, spacecraft formation flying \cite{wang20126_spacecraft_formation_1,gan2008dual_spacecraft_formation_2}, and computer vision.

In recent years, Koopman operator has drawn increasing attention among the controls and robotics community 
\cite{susuki2011nonlinear_koopman_6,surana2018_model_reduction_koopman_koopman_4,budivsic2012applied_koopman_theory_1}. Identification of Koopman-invariant subspaces using neural networks has been explored in \cite{lusch2018deep_koopman_neural_1,otto2019linearly_koopman_neural_2,yeung2019learning_koopman_neural_4} and using data-driven approaches in \cite{haseli2020fast,haseli2021learning,haseli2021parallel}. Extensions of these results / methods to controlled systems have also been explored for robotic applications \cite{abraham2019active_koopman_theory_tro,bruder2020data_soft_robots_koopman,mamakoukas2021derivative_koopman_robotics,mamakoukas2019local_koopman_robotics}, control synthesis \cite{folkestad2020data_koopman_control_synthesis_2,goswami2021bilinearization_koopman_control_synthesis_3,huang2018feedback_koopman}, aerospace applications \cite{chen2020koopman_attitude,zinage2021far_rendezvous_koopman}, power systems \cite{susuki2013nonlinear_power_systems_koopman_1,susuki2012nonlinear_power_systems_koopman_2}, control of PDEs \cite{peitz2019koopman_control_pdes} and climate research \cite{navarra2021estimation_climate_koopman}. 

The major challenge in using the Koopman operator based techniques for control and modeling of nonlinear dynamical systems is the characterization of the observable functions. State-of-the-art methods in the field, rely on heuristics or they try to learn these functions by using machine learning tools \cite{abraham2019active_koopman_theory_tro,korda2018linear_koopman_theory_automatica}.
The main advantage of using Koopman operator based techniques for modeling and control of dual quaternion based rigid body motion is two-fold. First, unlike linear models obtained through linearization about a fixed point whose accuracy is restricted in the vicinity of the latter point, the lifted linear model obtained by applying Koopman operator techniques provides an accurate description of the dynamics of the original system throughout a large subset of (if not the whole) the state space of the original system. Second, the availability of versatile and robust tools for analysis and control from the theory of linear systems make the analysis and control of the rigid body motion much easier.

\textit{Main contributions:} In this paper, we provide a systematic method to derive and construct the observable functions for the rigid body dynamics based on the dual quaternion representation. We show that these observables are functions of the dual quaternions which can form a sequence of functions. We prove that the latter sequence converges pointwise to the zero function. This result essentially allow us to truncate the proposed sequence of observables and obtain a finite-dimensional linear approximation of the rigid body dynamics which is sufficiently accurate for modeling and control design purposes. Further, we use these observables to design a data-driven Koopman based LQR controller for setpoint tracking. Through numerical simulations, we compare the efficacy of the proposed linear controller with a nonlinear controller \cite{nuno_filipe2013simultaneous_2} and show that our controller shows equivalent performance and is able to steer the rigid body to the desired state. To the best of the authors' knowledge, this is the first paper which proposes a Koopman operator framework for modeling and control of rigid body dynamics based on the dual quaternion representation.

{\textit{Structure of the paper:}} The organization of the paper is as follows. In Section \ref{sec:preliminaries}, we introduce the quaternion and dual quaternion algebra followed by the nonlinear rigid body dynamics represented in terms of dual quaternions. Subsequently, we provide an overview of the Koopman operator. In Section \ref{sec:observale_functions}, we provide 
the derivation of the set of Koopman based observables and then construct a lifted linear state space model based on these observables using EDMD. In Section \ref{sec:control_design}, we design a data-driven Koopman based LQR controller for the rigid body dynamics. Numerical simulations are presented in Section \ref{sec:numeical_simulations}, and finally Section \ref{sec:conclusions} presents concluding remarks.

\section{Preliminaries}\label{sec:preliminaries}
\subsection{Quaternion Algebra}
A quaternion $q$ can be represented by a pair $(\Bar{q},q_4)$ where $\Bar{q}\in\mathbb{R}^3$ is known as its vector part and $q_4\in\mathbb{R}$ as its scalar part. The set of quaternions is denoted by $\mathbb{Q}$. Some of the basic operations on quaternions are given below \cite{nuno_filipe2013rigid_nuno_rigid,nuno_filipe2013simultaneous_2}:
\begin{align}
&\text {Addition: } a+b=\left(\Bar{a}+\Bar{b}, a_{4}+b_{4}\right),\nonumber\\
&\text {Multiplication by a scalar: } \lambda a=\left(\lambda \Bar{a}, \lambda a_{4}\right),\nonumber\\
&\text {Multiplication: } a b=\left(a_{4} \Bar{b}+b_{4} \Bar{a}+\Bar{a} \times \Bar{b}, a_{4} b_{4}-\Bar{a} \cdot \Bar{b}\right),\nonumber \\
&\text {Conjugation: } a^{\star}=\left(-\Bar{a}, a_{4}\right),\nonumber \\
&\text {Dot product: } a \cdot b=\frac{1}{2}\left(a^{\star} b+b^{\star} a\right)=\frac{1}{2}\left(a b^{\star}+b a^{\star}\right), \nonumber\\
&\quad\quad\quad\quad\quad\quad\quad\;\;=\left(\Bar{0}, a_{4} b_{4}+\Bar{a} \cdot \Bar{b}\right),\nonumber \\
&\text {Cross product: } a \times b=\frac{1}{2}\left(a b-b^{\star} a^{\star}\right), \nonumber\\
&\quad\quad\quad\quad\quad\quad\quad\quad\quad=\left(b_{4} \Bar{a}+a_{4} \Bar{b}+\Bar{a} \times \Bar{b}, 0\right), \nonumber\\
&\text {Norm: }\|a\|^{2}=a a^{\star}=a^{\star} a=a \cdot a=\left(\Bar{{0}}, a_{4}^{2}+\Bar{a} \cdot \Bar{a}\right),\nonumber 
\end{align}
where $a,b\in\mathbb{Q}$.
\subsection{Dual Quaternion Algebra}
A dual quaternion $\widehat{q}$ can be represented by a pair $({q}_r,{q}_d)$ where ${q}_r,{q}_d\in\mathbb{Q}$ and ${q}_r$ and ${q}_d$ are the real and dual parts of $\widehat{q}$, respectively. Let the set of dual quaternions be denoted by $\mathbb{D}$. The dual quaternion $\widehat{q}$ can also be represented as $\widehat{q}=q_r+\epsilon q_d$ where $\epsilon^2=0$ and $\epsilon\neq0$. A list of some basic operations on dual quaternions are given below \cite{nuno_filipe2013rigid_nuno_rigid,nuno_filipe2013simultaneous_2}:
\begin{align}
&\text {Addition: } \widehat{a}+\widehat{b}=\left(a_{r}+b_{r}\right)+\epsilon\left(a_{d}+b_{d}\right),\nonumber\\
&\text {Multiplication by a scalar: } \lambda \widehat{a}=\left(\lambda a_{r}\right)+\epsilon\left(\lambda a_{d}\right),\nonumber\\
&\text {Multiplication: } \widehat{a} \widehat{b}=\left(a_{r} b_{r}\right)+\epsilon\left(a_{r} b_{d}+a_{d} b_{r}\right),\nonumber \\
&\text {Conjugation: } \widehat{a}^{\star}=a_{r}^{\star}+\epsilon a_{d}^{\star},\nonumber \\
 &\text {Swap: } \widehat{a}^{\mathrm{s}}=a_{d}+\epsilon a_{r},\nonumber \\
&\text {Dot product: } \widehat{a} \cdot \widehat{b}=\frac{1}{2}\left(\widehat{a}^{\star} \widehat{b}+\widehat{b}^{\star} \widehat{a}\right)=\frac{1}{2}\left(\widehat{a} \widehat{b}^{\star}+\widehat{b} \widehat{a}^{\star}\right)\nonumber\\
&\quad\quad\quad\quad\quad\quad\quad\quad =a_{r} \cdot b_{r}+\epsilon\left(a_{d} \cdot b_{r}+a_{r} \cdot b_{d}\right),\nonumber\\
&\text {Cross product: } \widehat{a} \times \widehat{b}=\frac{1}{2}\left(\widehat{a} \widehat{b}-\widehat{b}^{\star} \widehat{a}^{\star}\right)\nonumber \\
&\quad\quad\quad\quad\quad\quad\quad\quad\quad=a_{r} \times b_{r}+\epsilon\left(a_{d} \times b_{r}+a_{r} \times b_{d}\right),\nonumber \\
&\text {Circle product: }\widehat{a}\circ\widehat{b}=a_r\cdot b_r+a_d\cdot b_d,  \nonumber\\
&\text {Norm: }\|\widehat{a}\|^{2}=a_r\cdot a_r+a_d\cdot a_d , \nonumber\\
&\text {Dual norm: }\|\widehat{a}\|_{d}^{2}=\widehat{a} \widehat{a}^{\star}=\widehat{a}^{\star} \widehat{a}=\widehat{a} \cdot \widehat{a} \nonumber\\
&\quad\quad\quad\quad\quad\quad\quad\;\;=\left(a_{r} \cdot a_{r}\right)+\epsilon\left(2 a_{r} \cdot a_{d}\right),  \nonumber\\
& \text{Multiplication by matrix}: M \star \widehat{q}=\left(M_{11} \star q_{r}+M_{12} \star q_{d}\right)\nonumber\\
&\quad\quad\quad\quad\quad\quad\quad\quad\quad\quad+\epsilon\left(M_{21} \star q_{r}+M_{22} \star q_{d}\right),\nonumber
\end{align}
with $\widehat{a},\widehat{b}\in\mathbb{D}$ and $M= \Big[\begin{smallmatrix} M_{11} & M_{12} \\
M_{21} & M_{22} \end{smallmatrix} \Big]\in\mathbb{R}^{8\times 8}$, where $M_{11}, M_{12}, M_{21}, M_{22} \in \mathbb{R}^{4 \times 4}$.
The following lemma will be used in the subsequent analysis
\begin{lemma}
\normalfont If $\widehat{a},\widehat{b},\widehat{c}\in\mathbb{D}$, then $(\widehat{a}\widehat{b})\widehat{c}=\widehat{a}(\widehat{b}\widehat{c})$.
\end{lemma}
\begin{proof}
Using the multiplication property of quaternions and dual quaternions, we have
\begin{align}
    (\widehat{a}\widehat{b})\widehat{c}&=[a_rb_r+\epsilon(a_rb_d+a_db_r)](c_r+\epsilon c_d)\nonumber\\
    &=a_rb_rc_r+\epsilon[a_rb_rc_d+(a_rb_d+a_db_r)c_r]\nonumber\\
    &=a_rb_rc_r+\epsilon[a_rb_rc_d+a_rb_dc_r+a_db_rc_r],\\
  \widehat{a}(\widehat{b}\widehat{c})&=(a_r+\epsilon a_d)[b_rc_r+\epsilon(b_rc_d+b_dc_r)]\nonumber\\
  &=a_rb_rc_r+\epsilon[a_rb_rc_r+a_rb_dc_r+a_db_rc_r].
\end{align}
Hence $(\widehat{a}\widehat{b})\widehat{c}=\widehat{a}(\widehat{b}\widehat{c})$.
\end{proof}
\subsection{Kinematics of rigid body in terms of dual quaternions}
The kinematics of the rigid body in terms of the dual quaternions can be written as follows:
\begin{align}
    \dot{\widehat{q}}=\frac{1}{2}\widehat{q}\widehat{\omega}^B=\frac{1}{2}\widehat{\omega}^E\widehat{q},
    \label{eqn:kinematics_dual}
\end{align}
where $\widehat{\omega}^B=\omega^B+\epsilon v^B$, $\widehat{\omega}^E=\omega^E+\epsilon v^E$, $\widehat{q}=q_r+\epsilon q_d=q +\epsilon \frac{1}{2}q{{t}}^B$, $q$ is the rotation quaternion, ${{t}}^B=(\Bar{t}^B,0)$, $\Bar{t}^B$ is the translation vector in the body frame and superscript $B$ and $E$ denotes the body and inertial frame respectively. Further, $\omega^B=(\Bar{\boldsymbol{\omega}},0)\in\mathbb{Q}$ and $v^B=(\Bar{\boldsymbol{v}},0)\in\mathbb{Q}$ where $\Bar{\boldsymbol{\omega}}\in\mathbb{R}^3$ and $\Bar{\boldsymbol{v}}\in\mathbb{R}^3$ are the angular and linear velocities of the rigid body in the body frame respectively.
\subsection{Dynamics of rigid body in terms of dual quaternions}
The rigid body dynamics in terms of dual quaternions can be written as follows \cite{nuno_filipe2013rigid_nuno_rigid}:
\begin{align}
    M^B\star(\dot{\widehat{\omega}}^B)^s=\widehat{F}^B-\widehat{\omega}^B\times(M^B\star(({\widehat{\omega}}^B)^s))
    \label{eqn:dynamical_original}
\end{align}
where $(\cdot)^s$ denotes the swap operation performed on $(\cdot)$, $M^B$ is the dual inertia matrix, $\widehat{F}^B=F^B+\epsilon\tau^B$ is the dual force applied to the center of mass of the body, $F^B=(\Bar{F}^B,0)$ and $\tau=(\Bar{\tau},0)$. Consider a modified control input $\widehat{u}=\widehat{F}^B-\widehat{\omega}^B\times(M^B\star(({\widehat{\omega}}^B)^s))$. Then \eqref{eqn:dynamical_original} becomes
\begin{align}
    \dot{\widehat{\omega}}^B=(M^{B(-1)}\star\widehat{u})^s,
    \label{eqn:modified_control_input}
\end{align}
where $M^B\in\mathbb{R}^{8\times 8}$ is a block matrix
\begin{align}
    M^{B}=\left[\begin{array}{cccc}
m I_{3} & 0_{3 \times 1} & 0_{3 \times 3} & 0_{3 \times 1} \\
0_{1 \times 3} & 1 & 0_{1 \times 3} & 0 \\
0_{3 \times 3} & 0_{3 \times 1} & \Bar{I}^{B} & 0_{3 \times 1} \\
0_{1 \times 3} & 0 & 0_{1 \times 3} & 1
\end{array}\right],
\end{align}
where $\Bar{I}^B\in\mathbb{R}^{3\times 3}$ is positive definite matrix and is the mass moment of inertia and $m$ is the mass of the body. 
\begin{remark}
\normalfont A dual quaternion $\widehat{q}$ can also be represented in a vector form as follows:
\begin{align}
    \widehat{q}=[\bar{q}_r\;\;q_{r4}\;\;\bar{q}_d\;\;q_{d4}]^\mathrm{T}.\nonumber
\end{align}
Further, in order to keep the notation simple, the superscript $B$ for body frame will be dropped.
\end{remark}

\subsection{Koopman Operator}
In this section, we briefly review the key concepts from the Koopman operator theory. To this end, 
consider a continuous-time nonlinear dynamical system given by:
\begin{align}
   \dot{\boldsymbol{x}}=f(\boldsymbol{x}) 
   \label{eqn:dynamical_system}
\end{align}
where $\boldsymbol{x}\in \mathbb{R}^n $ and the vector field $f$ is assumed to satisfy regularity assumptions that ensure existence and uniqueness of solutions. 
Let $\mathcal{O}$ be the set of observables $\psi:\mathbb{R}^n\rightarrow \mathbb{C}$ where $\psi$ is an element of an infinite-dimensional Hilbert space and $\mathbb{C}$ belongs to the set of complex numbers. The Koopman operator $\mathcal{K}_t:\mathcal{O}\rightarrow\mathcal{O}$ associated with system \eqref{eqn:dynamical_system} is defined as follows:
\begin{align}
    [\mathcal{K}_t\psi(\boldsymbol{x})]=\mathcal{K}_t(\psi(\boldsymbol{x})).
\end{align}
Although the underlying dynamics is, in general, nonlinear, the Koopman operator is a linear infinite dimensional operator which acts on the space of observables. In particular, the following holds true:
\begin{align}
    [\mathcal{K}_t(\alpha\psi_1(\boldsymbol{x})+\beta\psi_2(\boldsymbol{x}))]=\alpha[\mathcal{K}_t\psi_1](\boldsymbol{x})+\beta[\mathcal{K}_t\psi_2](\boldsymbol{x}).\nonumber
\end{align}
where $\psi_1$, $\psi_2\in\mathcal{O}$ and $\alpha$, $\beta\in\mathbb{C}$.
 A Koopman eigenfunction $\psi_\lambda(\boldsymbol{x})\in\mathcal{O}$ corresponding to an eigenvalue $\lambda\in\mathbb{C}$ is defined as follows:
\begin{align}
    [\mathcal{K}_t\psi_\lambda(\boldsymbol{x})]=\lambda\psi_\lambda(\boldsymbol{x}).\nonumber
\end{align}
In other words, $\psi_\lambda(\boldsymbol{x})$ satisfies the following differential equation:
\begin{align}
    \dot{\psi}_\lambda(\boldsymbol{x})=\lambda\psi_\lambda(\boldsymbol{x})\nonumber
\end{align}
For instance, $\psi_\lambda({x})=e^{\frac{\lambda}{(1-n) a} {x}^{1-n}}$ is the eigenfunction to the polynomial (scalar) system $\dot{{x}}=a{x}^n$ with corresponding eigenvalue $\lambda$. For a controlled system of the form
\begin{align}
    \dot{\boldsymbol{x}}=f(\boldsymbol{x})+B\boldsymbol{u},\nonumber
\end{align}
with input matrix $B\in\mathbb{R}^{n\times m}$ and control input $\boldsymbol{u}\in\mathbb{R}^m$, the dynamics of the Koopman eigenfunctions become
\begin{align}
    \dot{\psi_\lambda}(\boldsymbol{x})=\lambda \psi_\lambda(\boldsymbol{x})+\nabla \psi_\lambda(\boldsymbol{x})  {B \boldsymbol{u}}
\end{align}
The Koopman operator $\mathcal{K}_d$ for discrete nonlinear system $\boldsymbol{\boldsymbol{x}}_{k+1}=h(\boldsymbol{\boldsymbol{x}}_k)$ can be written in terms of $\mathcal{K}$ and the sampling time $T$ as $ \mathcal{K}=\mathrm{log}(\mathcal{K}_d)/T$. Consequently,
\begin{align}
    \psi(\boldsymbol{x}_{k+1})=\mathcal{K}_d\psi(\boldsymbol{x}_k).\nonumber
\end{align}
%
In general it is not possible to find the set of finite Koopman eigenfunctions for any nonlinear dynamics. 
{In practice, one has to use a finite subspace approximation of the Koopman operator $\bar{\mathcal{K}}_d\in\mathbb{R}^N\times\mathbb{R}^N$ which acts on a subspace $\mathcal{S}\subset\mathcal{O}$ (recall that the Koopman operator $\mathcal{K}_d$ is infinite dimensional, in general).} If the finite set of observables are given by $\boldsymbol{z}(\boldsymbol{x})=[\psi_1(\boldsymbol{x})\;\;\psi_2(\boldsymbol{x})\dots\psi_N(\boldsymbol{x})]^\mathrm{T}\in\mathbb{R}^N$, the following approximation holds true
\begin{align}
    \boldsymbol{z}(\boldsymbol{x}_{k+1})\approx\Bar{\mathcal{K}}_d\boldsymbol{z}(\boldsymbol{x}_k)\nonumber
\end{align}
Given the data $\mathcal{D}=\{\boldsymbol{x}_k\}_{k=0}^{d}$, $\bar{\mathcal{K}}_d$ can be computed by solving the following least squares minimization problem:
\begin{align}
    \text{min}_{\bar{\mathcal{K}}_d}~ \|  \boldsymbol{z}(\boldsymbol{x}_{k+1})-\bar{\mathcal{K}}_d\boldsymbol{z}(\boldsymbol{x}_k)\|_2^2.
\end{align}
Consider the discrete-time controlled system $\boldsymbol{x}_{k+1}=h(\boldsymbol{x}_k,\boldsymbol{u}_k)$. Then, the Koopman operator $\mathcal{K}_d$ over the extended state space $\mathbb{G}:\mathcal{X}\times\mathcal{U}$ and observable $g(\boldsymbol{x}_k,\boldsymbol{u}_k)=[\boldsymbol{z}\;\;v(\boldsymbol{x}_k,\boldsymbol{u}_k)]^\mathrm{T}$ can be defined as follows \cite{abraham2019active_koopman_theory_tro}:
\begin{align}
   g(\boldsymbol{x}_{k+1},\boldsymbol{u}_{k+1})=\mathcal{K}_dg(\boldsymbol{x}_k,\boldsymbol{u}_k). 
\end{align}

\begin{figure*}
\centering
\includegraphics[scale=0.68]{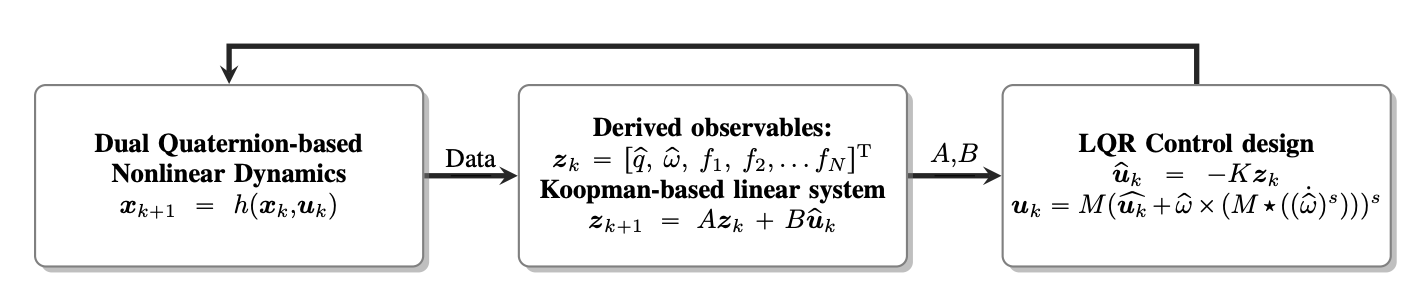}
 \caption{ Control of nonlinear rigid body motion based on dual quaternion representation using linear control design on the data-driven Koopman based lifted space linear dynamics.}
 \label{fig:feedback_diagram}
\end{figure*}
\section{Derivation of Koopman based observables\label{sec:observale_functions}}
In this section, we provide a systematic way to derive the observable functions for the continuous-time rigid body motion based on the dual quaternion representation. 

\begin{theorem}
\normalfont For the nonlinear system governed by \eqref{eqn:kinematics_dual} and \eqref{eqn:dynamical_original}, the lifted state space $\boldsymbol{z}$ is spanned by the following set of observable functions
\begin{align}
    \boldsymbol{z}=(\widehat{q},\;\;\widehat{\omega},\;\{\widehat{f}_k\}_0^\infty).
    \label{eqn:observables_list}
\end{align}
where $\widehat{f}_k=\widehat{q}{\widehat{\omega}^k}$.
\label{thm:observables}
\end{theorem}
\begin{proof}
Let $\widehat{{\omega}}^{B(k)}$ be defined as follows:
\begin{align}
   \widehat{{\omega}}^{B(k)}=\underbrace{(\widehat{{\omega}}^B(\widehat{{\omega}}^B(\dots(\widehat{{\omega}}^B\widehat{{\omega}}^B)}_{k\;\text{times}}\dots))). \nonumber
\end{align}
Let $f_1=\widehat{q}\widehat{\omega}^B$. Then,
\begin{align}
    \dot{f}_1&=\dot{\widehat{q}}\widehat{\omega}^B+\widehat{q}\dot{\widehat{\omega}}^B=\frac{1}{2}(\widehat{q}\widehat{\omega}^B)\widehat{\omega}^B+\widehat{q}\dot{\widehat{\omega}}^B
\nonumber\\
&=\frac{1}{2}\widehat{q}(\widehat{\omega}^B\widehat{\omega}^B)+\widehat{q}\dot{\widehat{\omega}}^B=\frac{1}{2}\widehat{q}\widehat{\omega}^{B(2)}+\widehat{q}\dot{\widehat{\omega}}^B,\nonumber
\end{align}
where ${\widehat{\omega}^{B(2)}}=\widehat{\omega}^B\widehat{\omega}^B$. Now let $f_2=\widehat{q}\widehat{\omega}^{B(2)}$. Then,
\begin{align}
     \dot{f}_2&=\dot{\widehat{q}}{\widehat{\omega}^{B(2)}}+\widehat{q}{\dot{\widehat{\omega}}^{B(2)}}
    =\frac{1}{2}(\widehat{q}\widehat{\omega}^B){\widehat{\omega}^{B(2)}}+\widehat{q}\dot{\widehat{\omega}}^B
\nonumber\\
&=\frac{1}{2}\widehat{q}(\widehat{\omega}^B{\widehat{\omega}^{B(2)}})+\widehat{q}\dot{\widehat{\omega}}^B=\frac{1}{2}\widehat{q}\widehat{\omega}^{B(3)}+\widehat{q}\dot{\widehat{\omega}}^B,   
\end{align}
where ${\widehat{\omega}^{B(3)}}=\widehat{\omega}^B(\widehat{\omega}^B\widehat{\omega}^B)$. Therefore, in general
\begin{align}
  \dot{f}_k=\frac{1}{2}f_{k+1}+\sum_{i=1}^k {\widehat{\omega}}^{B {(i-1)}}({M^{-1}}\star\tilde{\boldsymbol{u}}) {\widehat{\omega}}^{B{(k-i)}}. 
\end{align}
As $N\rightarrow\infty$, we obtain countably infinite collection of observables given by \eqref{eqn:observables_list}.
\end{proof}
In the following we derive the general expression for $\widehat{\omega}^{B(k)}$ which will be used in subsequent analysis. 

\subsection{General expression for $\widehat{\omega}^{B(k)}$ }
Consider the expression for ${\widehat{\omega}_1^{B}}:=\widehat{\omega}^B\widehat{\omega}^B$,
\begin{align}
 \widehat{\omega}^B\widehat{\omega}^B&=(\omega^B\omega^B)+\epsilon(\omega^Bv^B+v^B\omega^B).
\end{align}
Now, $(\omega^B\omega^B), \omega^Bv^B$ and $v^B\omega^B$ can be written as
\begin{align}
    \omega^B\omega^B=&(\Bar{\boldsymbol{\omega}}^B\times\Bar{\boldsymbol{\omega}}^B,-\Bar{\boldsymbol{\omega}}^B.\Bar{\boldsymbol{\omega}}^B)\nonumber\\
    =&(0,-|\Bar{\boldsymbol{\omega}}^B|^2),\nonumber\\
    \omega^Bv^B=&(\Bar{\boldsymbol{\omega}}^B\times\Bar{\boldsymbol{v}}^B,-\Bar{\boldsymbol{\omega}}^B.\Bar{\boldsymbol{v}}^B),\nonumber\\
     v^B\omega^B=&(\Bar{\boldsymbol{v}}^B\times\Bar{\boldsymbol{\omega}}^B,-\Bar{\boldsymbol{\omega}}^B.\Bar{\boldsymbol{v}}^B).\nonumber
\end{align}
Therefore,
\begin{align}
  \widehat{\omega}^B\widehat{\omega}^B=(0,-|\Bar{\boldsymbol{\omega}}^B|^2)+\epsilon(\Bar{0},-2\Bar{\boldsymbol{\omega}}^B.\Bar{\boldsymbol{v}}^B).  
\end{align}
Now consider the expression of $\widehat{\omega}^B\widehat{\omega}^B_1=\widehat{\omega}^B\widehat{\omega}^B\widehat{\omega}^B$ can be written as follows:
\begin{align}
\widehat{\omega}^B\widehat{\omega}^B_1=\omega^B\omega_1^B+\epsilon(\omega_1^Bv^B+v_1^B\omega^B) . 
\label{eqn:omega_omega1}
\end{align}
Now, $\omega^B\omega_1^B,\omega_1^Bv^B$ and $v_1^B\omega^B$ in \eqref{eqn:omega_omega1} are given as \begin{align}
  & \omega^B\omega_1^B=(-|\Bar{\boldsymbol{\omega}}^{B(2)}|\Bar{\boldsymbol{\omega}}^{B},0),\quad\omega_1^Bv^B=(-|\Bar{\boldsymbol{\omega}}^{B(2)}|\Bar{\boldsymbol{v}}^{B},0),\nonumber\\
  &v_1^B\omega^B=(-2(\Bar{\boldsymbol{\omega}}^B.\Bar{\boldsymbol{v}}^B)\Bar{\boldsymbol{\omega}}^B,0),\nonumber
\end{align}
Therefore
\begin{align}
  \widehat{\omega}_2^{B}:= \widehat{\omega}^B\widehat{\omega}^B_1=&(-|\Bar{\boldsymbol{\omega}}^{B(2)}|\Bar{\boldsymbol{\omega}}^{B},0)+\epsilon(-|\Bar{\boldsymbol{\omega}}^{B(2)}|\Bar{\boldsymbol{v}}^{B}-\nonumber\\&2(\Bar{\boldsymbol{\omega}}^B.\Bar{\boldsymbol{v}}^B)\Bar{\boldsymbol{\omega}}^B,0)\nonumber  
\end{align}
Now, the expression of $\widehat{\omega}_3^{B}:=\widehat{\omega}^B\widehat{\omega}_2^B$ is given by
\begin{subequations}
\begin{align}
   &\widehat{\omega}^B\widehat{\omega}_2^B=(\omega^B\omega_2^B)+\epsilon(\omega_2^Bv^B+v_2^B\omega^B)\\
   &\omega^B\omega_2^B=(0,-|\Bar{\boldsymbol{\omega}}^B|^4)\\
   & \omega_2^Bv^B=(|\Bar{\boldsymbol{\omega}}^B|^2(\Bar{\boldsymbol{\omega}}^B\times\Bar{\boldsymbol{v}}^B),-|\Bar{\boldsymbol{\omega}}^B|^2((\Bar{\boldsymbol{\omega}}^B.\Bar{\boldsymbol{v}}^B)))\\
   & v_2^B\omega^B=(-|\Bar{\boldsymbol{\omega}}^B|^2(\Bar{\boldsymbol{\omega}}^B\times\Bar{\boldsymbol{v}}^B),-|\Bar{\boldsymbol{\omega}}^B|^2\Bar{\boldsymbol{\omega}}^B.\Bar{\boldsymbol{v}}^B-2(\Bar{\boldsymbol{\omega}}^B.\Bar{\boldsymbol{v}}^B)^2)
\end{align}
\label{eqn:omega_3_hat}
\end{subequations}
Therefore, from \eqref{eqn:omega_3_hat}
\begin{align}
  \widehat{\omega}_3^{B}= (0,-|\Bar{\boldsymbol{\omega}}^B|^4)+\epsilon(0,-2|\Bar{\boldsymbol{\omega}}^B|^2\Bar{\boldsymbol{\omega}}^B.\Bar{\boldsymbol{v}}^B-2(\Bar{\boldsymbol{\omega}}^B.\Bar{\boldsymbol{v}}^B)^2)\nonumber
\end{align}
Again, let $\widehat{\omega}_4^B:=\widehat{\omega}^B\widehat{\omega}_3^B$. Then,
\begin{subequations}
\begin{align}
    &\widehat{\omega}^B\widehat{\omega}_3^B=(\omega^B\omega_3^B)+\epsilon(\omega_3^Bv^B+v_3^B\omega^B),\\ 
    & \omega_3^B\omega^B=(-|\Bar{\boldsymbol{\omega}}^B|^4\Bar{\boldsymbol{\omega}}^B,0),\\
    & \omega_3^Bv^B=(-|\Bar{\boldsymbol{\omega}}^B|^4\Bar{\boldsymbol{v}}^B,0),\\
    & v_3^B\omega^B=(-2|\Bar{\boldsymbol{\omega}}^B|^2(\Bar{\boldsymbol{\omega}}^B.\Bar{\boldsymbol{v}}^B)\Bar{\boldsymbol{\omega}}^B-2(\Bar{\boldsymbol{\omega}}^B.\Bar{\boldsymbol{v}}^B)^2\Bar{\boldsymbol{\omega}}^B,0).
\end{align}
\label{eqn:omega_4_hat}
\end{subequations}
Therefore, from \eqref{eqn:omega_4_hat}
\begin{align}
    \widehat{\omega}_4^B:=\widehat{\omega}^B\widehat{\omega}_3^B=&(-|\Bar{\boldsymbol{\omega}}^B|^4\Bar{\boldsymbol{\omega}}^B,0)+\epsilon(-|\Bar{\boldsymbol{\omega}}^B|^4\Bar{\boldsymbol{v}}^B-\nonumber\\&2|\Bar{\boldsymbol{\omega}}^B|^2(\Bar{\boldsymbol{\omega}}^B.\Bar{\boldsymbol{v}}^B)\Bar{\boldsymbol{\omega}}^B
    -2(\Bar{\boldsymbol{\omega}}^B.\Bar{\boldsymbol{v}}^B)^2\Bar{\boldsymbol{\omega}}^B,0)\nonumber
\end{align}
Hence, the value of $\widehat{{\omega}}^{B(k)}$ can be written as follows:
\begin{itemize}
    \item Case 1: $k$ is odd
    \begin{align}
     \widehat{{\omega}}^{B(k)}  =& (|\Bar{\boldsymbol{\omega}}^B|^{k-1}\Bar{\boldsymbol{\omega}}^B,0)+\epsilon(-|\Bar{\boldsymbol{\omega}}^B|^{(k-1)}\Bar{\boldsymbol{v}}^B-\nonumber\\
&2\sum_{i=1}^{\frac{k-1}{2}}|\Bar{\boldsymbol{\omega}}^B|^{(k-1-2i)}(\Bar{\boldsymbol{\omega}}^B.\Bar{\boldsymbol{v}}^B)^i\Bar{\boldsymbol{\omega}}^B,0)  \nonumber
    \end{align}
    \item Case 2: $k$ is even
    \begin{align}
       \widehat{{\omega}}^{B(k)}   =(\bar{0},-|\Bar{\boldsymbol{\omega}}^B|^{k})+\epsilon(\bar{0},-2\sum_{i=1}^{k/2}|\Bar{\boldsymbol{\omega}}^B|^{(k-2i)}(\Bar{\boldsymbol{\omega}}^B.\Bar{\boldsymbol{v}}^B)^i)\nonumber
    \end{align}
\end{itemize}
The following lemma will be used to prove the pointwise convergence of the observables to zero.
\begin{lemma}
\normalfont For any $\widehat{a}, \widehat{b} \in \mathbb{D},$ we have
\begin{align}
\|\widehat{a} \widehat{b}\| \leq \sqrt{3 / 2}\|\widehat{a}\|\|\widehat{b}\|.
\end{align}
\label{lemma:inequality}
\end{lemma}
\begin{proof}
Refer to the proof of Lemma 1 from \cite{nuno_filipe2013rigid_nuno_rigid}.
\end{proof}
\begin{assumption}
\normalfont We assume that the maximum angular and linear velocities of the rigid body are constrained and are known a-priori. In other words, there exists some $\Bar{\boldsymbol{\omega}}_{0}$ and $\Bar{\boldsymbol{v}}_{0}$ such that
\begin{align}
 {\omega}_0>\underset{\Bar{\boldsymbol{\omega}}}{\texttt{max}}(|\Bar{\boldsymbol{\omega}}|),\;\quad  {v}_0>\underset{\Bar{\boldsymbol{v}}}{\texttt{max}}(|\Bar{\boldsymbol{v}}|). \nonumber
\end{align}
\end{assumption}
Now, let us consider the normalized angular and linear velocities which are defined as follows:
\begin{align}
    \|\tilde{\Bar{\boldsymbol{\omega}}}\|=\frac{\|\Bar{\boldsymbol{\omega}}\|}{\texttt{max}(\{\omega_0,v_0\})}<1,\;\;  \|\tilde{\Bar{\boldsymbol{v}}}\|=\frac{\|\Bar{\boldsymbol{v}}\|}{\texttt{max}(\{\omega_0,v_0\})}<1\nonumber
\end{align}
Next, we define the modified observable function $\widehat{f}_k$ as 
\begin{align}
    \widehat{f}_k=\widehat{q}\widetilde{\widehat{\omega}}^k
\end{align}
where $\tilde{\widehat{\omega}}=(\tilde{\Bar{\boldsymbol{\omega}}},0)+\epsilon(\tilde{\Bar{\boldsymbol{v}}},0)$. The expression of $\widehat{f}_k$ can be written in terms of $f_k$ as follows:
\begin{align}
    \widehat{f}_k=(\widehat{\omega}_0)^kf_k,
\end{align}
where $\widehat{\omega}_0=(0,\texttt{max}(\{\omega_0,v_0\}))+\epsilon(\Bar{0},0)$. The linear dynamics in the lifted space can then be written as follows:
\begin{align}
    \dot{\widehat{f}}_k=\widehat{\omega}_0\widehat{f}_{k+1}+\widehat{q}\sum_{i=1}^k \tilde{\widehat{\omega}}^{B {(i-1)}}({M^{-1}}\star\widetilde{\boldsymbol{u}})^s \tilde{\widehat{\omega}}^{B{(k-i)}} 
\end{align}
In addition, let
\begin{align}
    B_k:=\sum_{i=1}^k \tilde{\widehat{\omega}}^{B {(i-1)}}({M^{-1}}\star\tilde{\boldsymbol{u}})^s \tilde{\widehat{\omega}}^{B{(k-i)}} \label{eqn:bk_expression}
\end{align}
Next we consider the sets $\mathcal{D}_{\widehat{\omega}}$ and $\mathcal{D}_{\widehat{v}}$ where
\begin{align}
    \mathcal{D}_{\widehat{\omega}}:=\{\widetilde{\widehat{\omega}}:\|\tilde{\Bar{\boldsymbol{\omega}}}\|<1\},\;\;
     \mathcal{D}_{\widehat{v}}:=\{\widetilde{\widehat{v}}:\|\tilde{\Bar{\boldsymbol{v}}}\|<1\}\nonumber
\end{align}

\begin{lemma}
\normalfont For $k\in[2,N]_d$, the following holds true:
\begin{align}
  \underset{\texttt{max}(\{\omega_0,v_0\})\to\infty}{\lim} B_k=\widehat{0},\quad \underset{k\to\infty}{\lim} B_k=\widehat{0}
\end{align}
\end{lemma}

\begin{proof}
Since $\tilde{\widehat{\omega}}=[(\Bar{\boldsymbol{\omega}},0)+\epsilon(\Bar{\boldsymbol{v}},0)]/\texttt{max}(\{\omega_0,v_0\})$. Therefore
\begin{align}
   \underset{\texttt{max}(\{\omega_0,v_0\})\to\infty}{\lim} \tilde{\widehat{\omega}}=(\Bar{\boldsymbol{0}},0)+\epsilon(\Bar{\boldsymbol{0}},0)
\end{align}
Subsequently from \eqref{eqn:bk_expression},
\begin{align}
  \underset{\omega_0\to\infty}{\lim} B_k =(\Bar{{0}},0)+\epsilon(\Bar{{0}},0)
\end{align}
Further using Lemma \ref{lemma:inequality},
\begin{align}
    \|B_k\|\leq \frac{3}{2}k\|\tilde{\widehat{\omega}}^{B{(k-i)}}\|\|({M^{-1}}\star\tilde{\boldsymbol{u}})^s\|
    \label{eqn:bk_norm}
\end{align}
Since $\underset{k\to\infty}{\lim}kx^k=0$ for $x<1$, taking limit on both sides of \eqref{eqn:bk_norm}
\begin{align}
    \underset{k\to\infty}{\lim} \|B_k\|=0.\nonumber
\end{align}
Consequently, $\underset{k\to\infty}{\lim} B_k=\widehat{0}$.
This completes the proof.
\end{proof}
\begin{remark}
\normalfont For higher value of $\omega_0$, $B_k$ (for all $k\in[2,N]_d$) can be approximated to be the zero dual number i.e. $B_k \approx(\Bar{\boldsymbol{0}},0)+\epsilon(\Bar{\boldsymbol{0}},0) $. In other words, as $k$ and $\omega_0$ increases, the dependence of the states on $B_k$ decreases. Thereafter, the lifted space linear dynamics can be approximated as follows:
\begin{subequations}
\begin{align}
       &\dot{\widehat{f}}_1=\widehat{\omega}_0\widehat{f}_{2}+\widehat{q}({M^{-1}}\star\tilde{\boldsymbol{u}})^s/\omega_0,\\
       &     \dot{\widehat{f}}_k=\widehat{\omega}_0\widehat{f}_{k+1}+\widehat{q}B_k,\quad \quad k\in[2,N]_d
       \label{eqn:approximated_lifted_space_dynamics}
\end{align}
\end{subequations}
\end{remark}
\begin{theorem}
\normalfont {For any $\boldsymbol{\omega}\in\mathcal{D}_{\widehat{\omega}}$ and $\boldsymbol{v}\in\mathcal{D}_{\widehat{v}}$, the sequences of functions $\widehat{f}_k$ and $\dot{\widehat{f}}_k$ converge pointwise to $\widehat{0}$, i.e.,
\begin{figure*}[]
 \captionsetup[subfigure]{justification=centering}
 \centering
 \begin{subfigure}{0.23\textwidth}
{\includegraphics[scale=0.22]{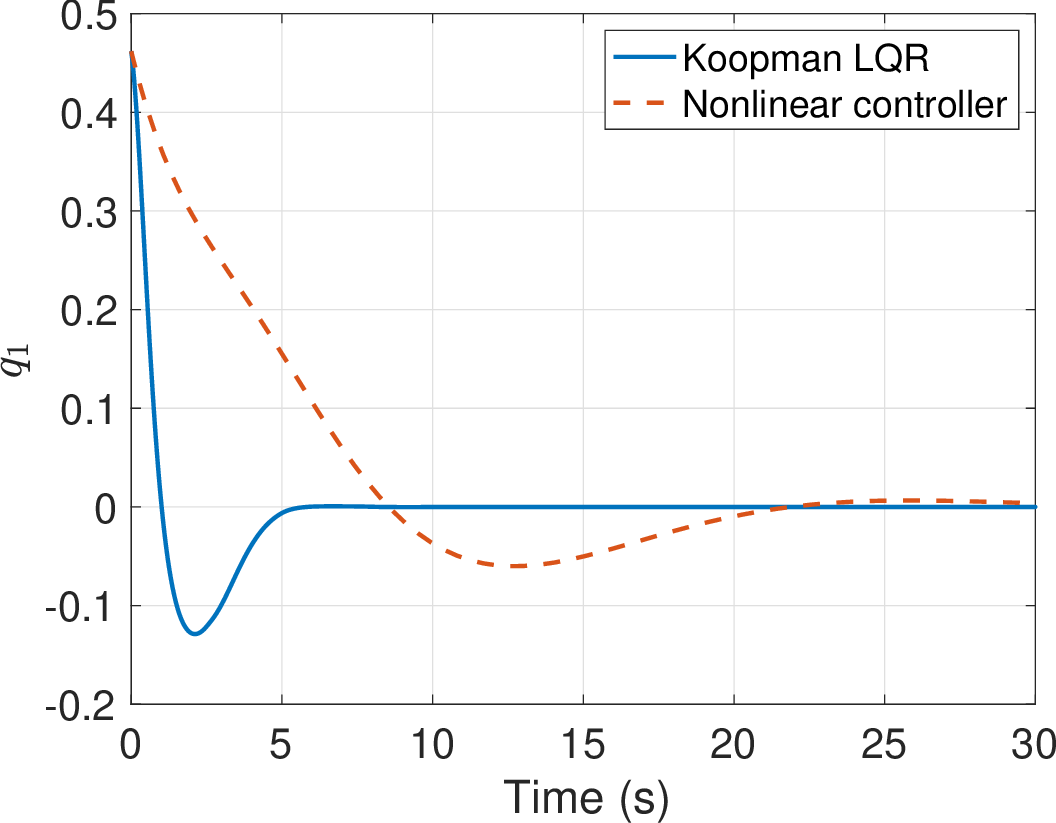}}
\caption{$q_1$}
\label{fig:}
 \end{subfigure}
\label{fig:}
 \begin{subfigure}{0.23\textwidth}
{\includegraphics[scale=0.22]{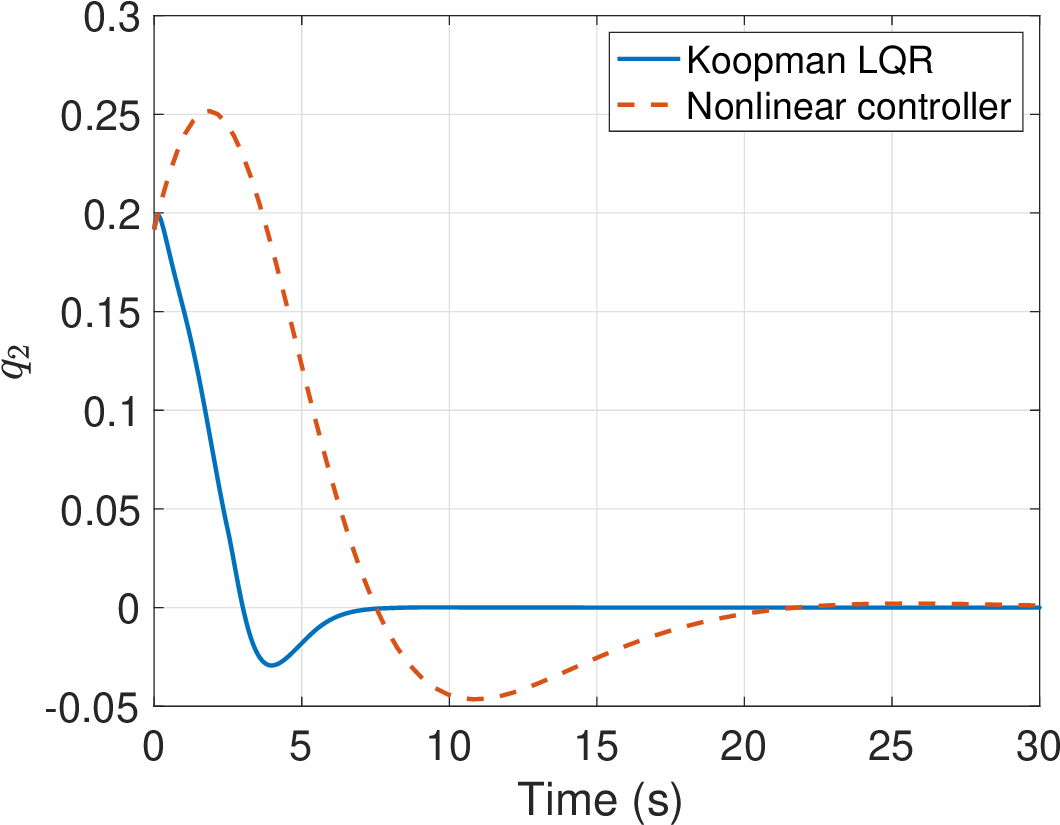}}
 \caption{$q_2$}
\label{fig:approximation_e_v}
 \end{subfigure}
\label{fig:}
 \begin{subfigure}{0.23\textwidth}
{\includegraphics[scale=0.22]{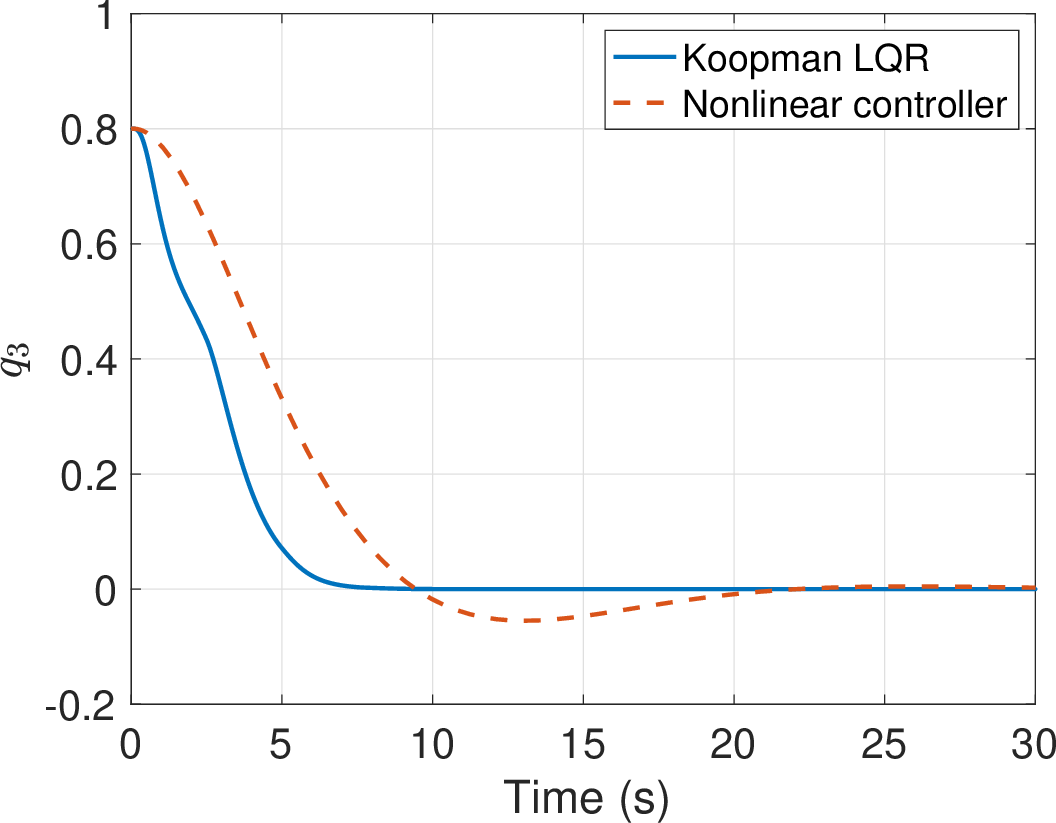}}
\caption{$q_3$}
\label{fig:}
\end{subfigure}
 \begin{subfigure}{0.23\textwidth}
{\includegraphics[scale=0.22]{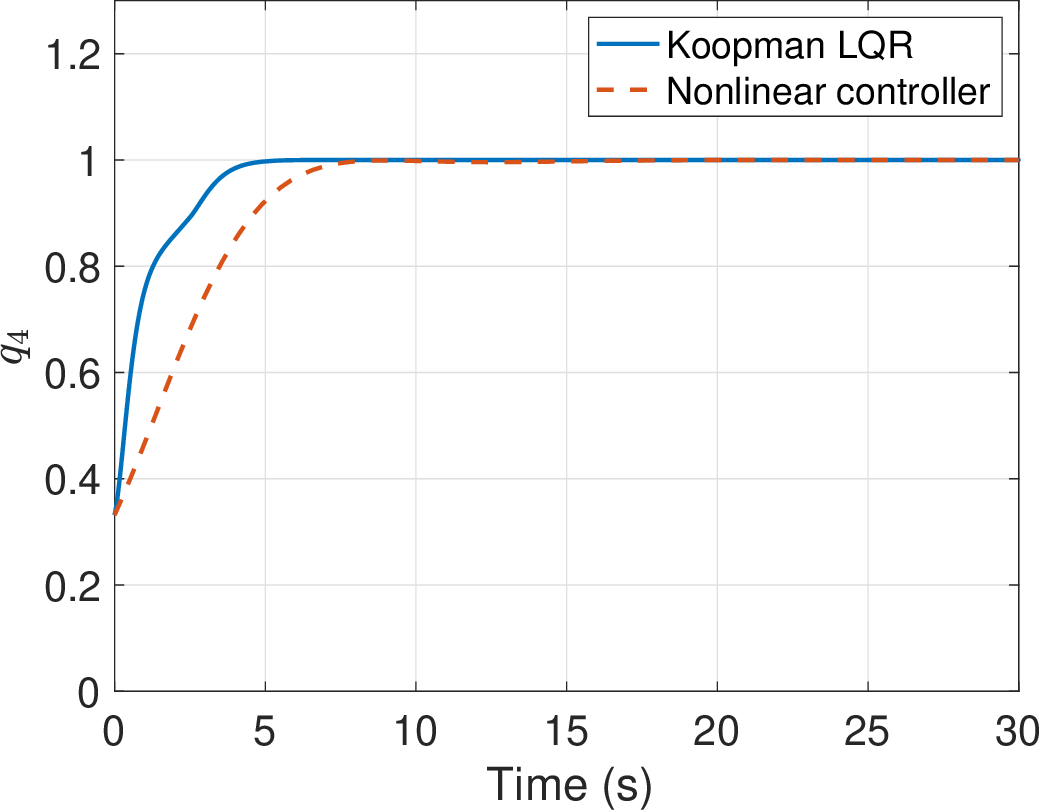}}
\caption{$q_4$}
\label{fig:}
\end{subfigure}
 \begin{subfigure}{0.23\textwidth}
{\includegraphics[scale=0.22]{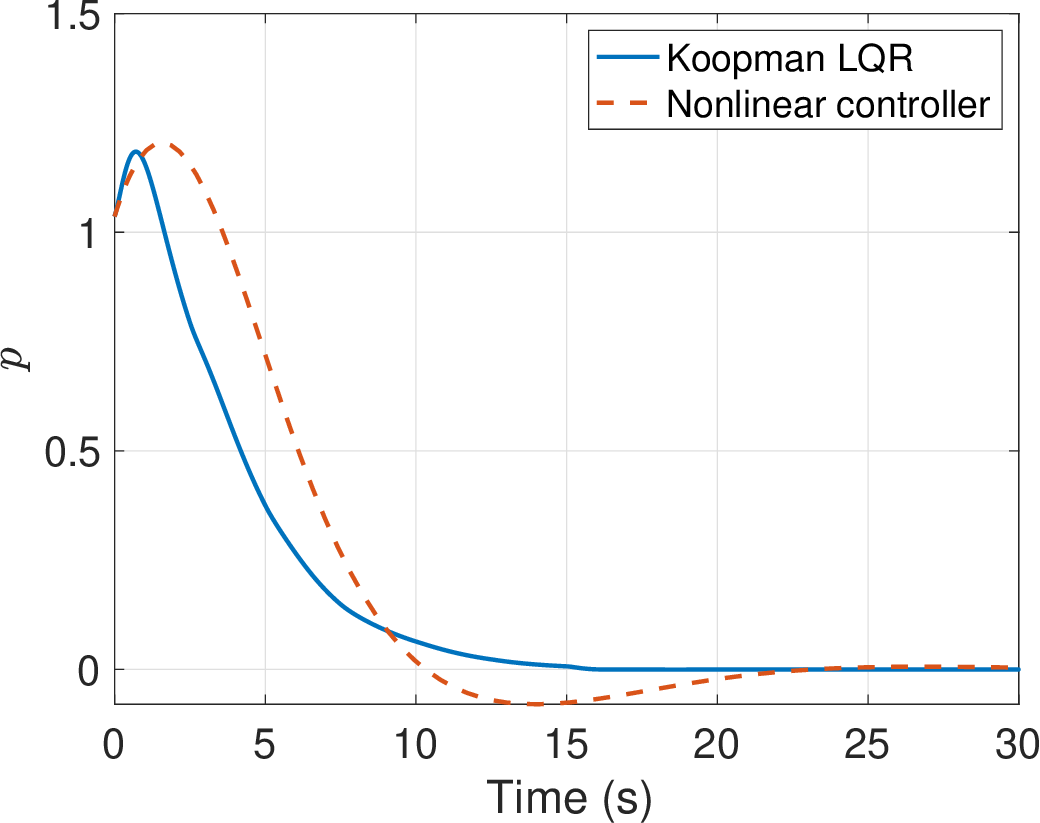}}
\caption{$q_5$}
\label{fig:}
 \end{subfigure}
\label{fig:}
 \begin{subfigure}{0.23\textwidth}
{\includegraphics[scale=0.22]{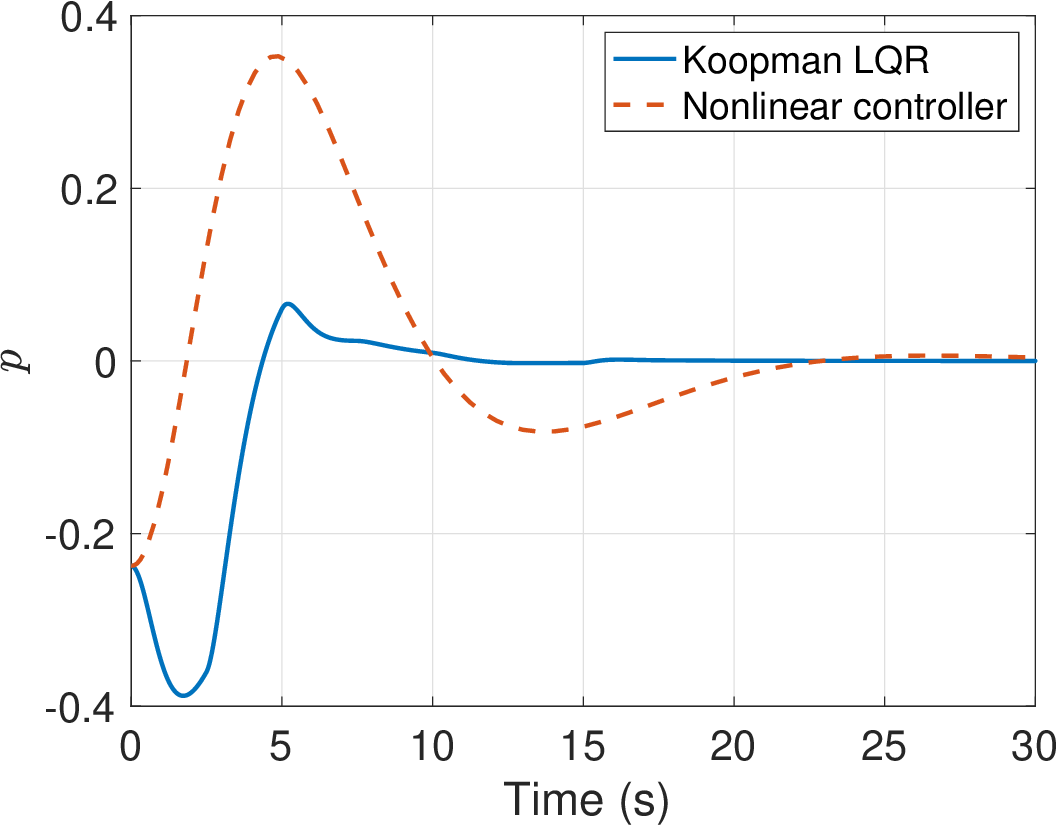}}
 \caption{$q_6$}
\label{fig:approximation_e_v}
 \end{subfigure}
\label{fig:}
 \begin{subfigure}{0.23\textwidth}
{\includegraphics[scale=0.22]{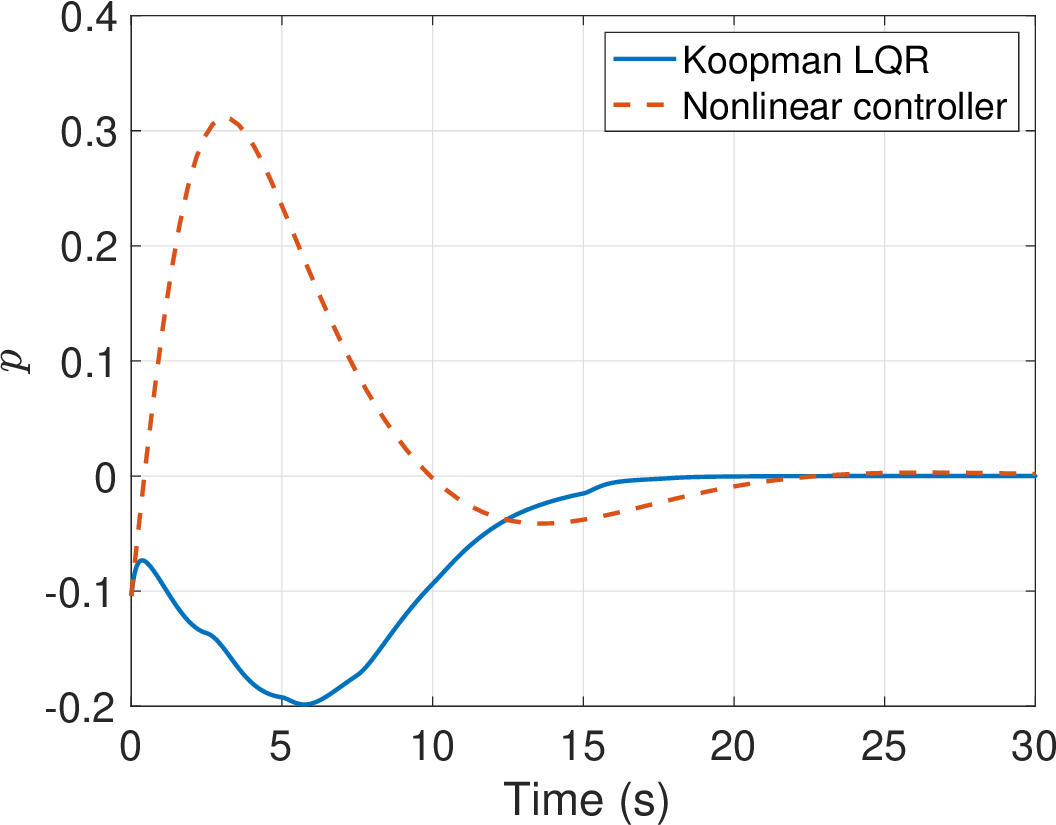}}
\caption{$q_7$}
\label{fig:}
\end{subfigure}
 \begin{subfigure}{0.23\textwidth}
{\includegraphics[scale=0.22]{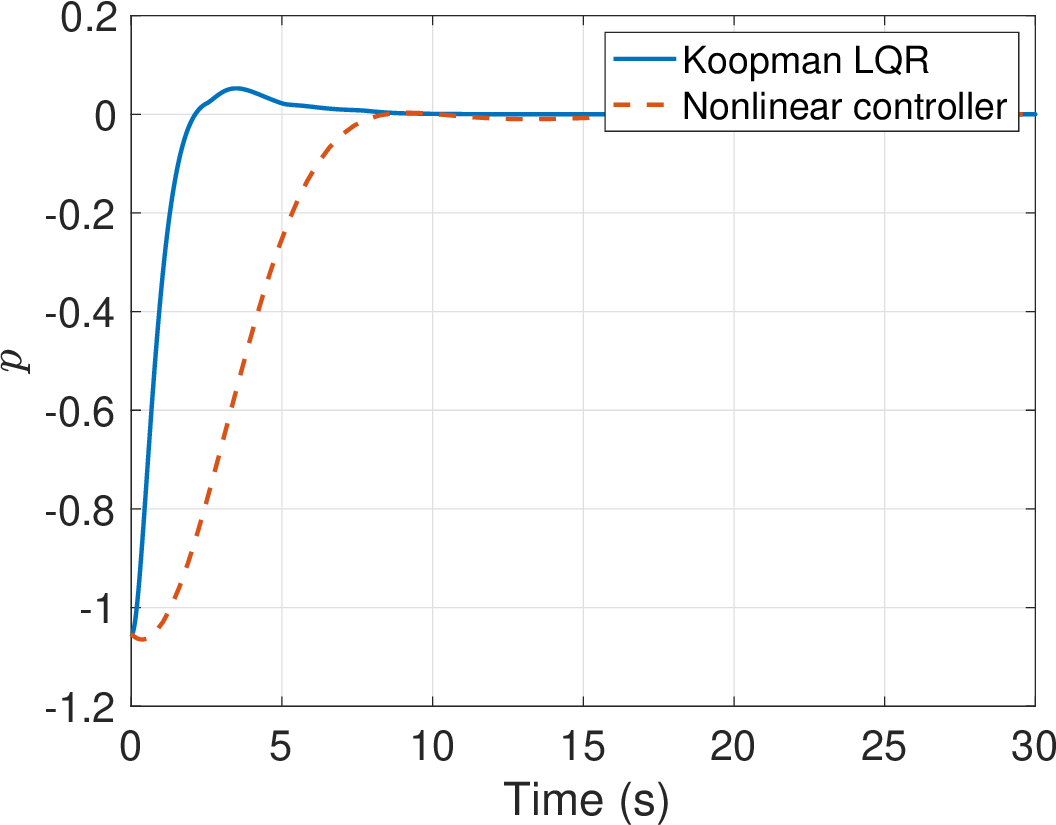}}
\caption{$q_8$}
\label{fig:}
\end{subfigure}
\caption{Evolution of the dual quaternion $\widehat{q}$ with time}
\label{fig:quaternion_plot}
\end{figure*}
\begin{align}
   &\underset{k\to\infty}{\lim}{\widehat{f}}_k=\widehat{{0}},\quad  \underset{k\to\infty}{\lim}{\dot{\widehat{f}}}_k=\widehat{{0}}\quad\forall\;\boldsymbol{\omega}\in\mathcal{D}_{\widehat{\omega}},\;\; \boldsymbol{v}\in\mathcal{D}_{\widehat{v}}.\nonumber
\end{align}}
\label{prop:f_k_hat}
\end{theorem}
\begin{proof}
Since $\|\tilde{\Bar{\boldsymbol{\omega}}}\|<1$ and $\|\tilde{\widetilde{\Bar{\boldsymbol{v}}}}\|<1$, we have
\begin{align}
    \underset{k\to\infty}{\lim}\|\widetilde{\Bar{\boldsymbol{\omega}}}^{B(k-1)}\|\widetilde{\Bar{\boldsymbol{\omega}}}=\mathbf{0},\quad  \underset{k\to\infty}{\lim}\|\widetilde{\Bar{\boldsymbol{\omega}}}^{B(k)}\|=0.
    \label{eqn:limit_omega_h_is_0}
\end{align}
In addition, we have
\begin{align}
  \sum_{i=1}^k\|\widetilde{\Bar{\boldsymbol{\omega}}}^B\|^{(k-2i)}(\widetilde{\Bar{\boldsymbol{\omega}}}^B.\widetilde{\Bar{\boldsymbol{v}}}^B)^i\leq\sum_{i=1}^k\|\widetilde{\Bar{\boldsymbol{\omega}}}^B\|^{(k-i)}\|\widetilde{\Bar{\boldsymbol{v}}}^B\|^{(i)}  
\end{align}

Now, using the formula for the sum of geometric series, we have
\begin{align}
\small
   \sum_{i=1}^k\|\widetilde{\Bar{\boldsymbol{\omega}}}^B\|^{(k-i)}\|\widetilde{\Bar{\boldsymbol{v}}}^B|^{(i)}&=\small\frac{\|\widetilde{\Bar{\boldsymbol{\omega}}}^B|^{(k-1)}\|\widetilde{\Bar{\boldsymbol{v}}}^B\|(1-(|\widetilde{\Bar{\boldsymbol{v}}}^B\|/\|\widetilde{\Bar{\boldsymbol{\omega}}}^B\|)^k)}{1-\|\widetilde{\Bar{\boldsymbol{v}}}^B\|/\|\widetilde{\Bar{\boldsymbol{\omega}}}^B\|}
   \label{eqn:summation}
\end{align}
Taking limits on both sides of \eqref{eqn:summation} gives
\begin{align}
    \underset{k\to\infty}{\lim}\sum_{i=1}^k\|\widetilde{\Bar{\boldsymbol{\omega}}}^B|^{(k-i)}\|\|\widetilde{\Bar{\boldsymbol{v}}}^B\|^{(i)}=0.
    \label{eqn:limit_sum_is_0}
\end{align}
Similarly, it can be shown that
\begin{align}
  \underset{k\to\infty}{\lim}\|\widetilde{\Bar{\boldsymbol{\omega}}}^B\|^{(k-1)}\widetilde{\Bar{\boldsymbol{v}}}^B+2\sum_{i=1}^k\|\widetilde{\Bar{\boldsymbol{\omega}}}^B\|^{(k-1-2i)}(\widetilde{\Bar{\boldsymbol{\omega}}}^B.\widetilde{\Bar{\boldsymbol{v}}}^B)^i\widetilde{\Bar{\boldsymbol{\omega}}}^B=\mathbf{0}
  \label{eqn:limit_sum_is_0_second}
\end{align}
Using \eqref{eqn:limit_omega_h_is_0}, \eqref{eqn:limit_sum_is_0} and \eqref{eqn:limit_sum_is_0_second}, we conclude that
\begin{align}
  \underset{k\to\infty}{\lim}\widehat{{\omega}}^{B(k)} =(\Bar{\boldsymbol{0}},0)+\epsilon(\Bar{\boldsymbol{0}},0). 
\end{align}
Since $\widehat{f}_k=\widehat{q}\tilde{\widehat{\omega}}^{B(k)}$, using Lemma 2, we have
\begin{align}
   \|\widehat{q} \tilde{\widehat{\omega}}^{B(k)}\| \leq \sqrt{3 / 2}\|\widehat{q}\|\|\tilde{\widehat{\omega}}^{B(k)}\| .
\end{align}
Now, since $\underset{k\to\infty}{\lim}\|\widetilde{\Bar{\boldsymbol{\omega}}}^B\|^{(k)} =0 $, we have
\begin{align}
 \underset{k\to\infty}{\lim} \|\widehat{q} \tilde{\widehat{\omega}}^{B(k)}\|\leq 0\implies \underset{k\to\infty}{\lim}\widehat{f}_k=\widehat{\boldsymbol{0}}.\nonumber
\end{align}
Taking limits on both sides of \eqref{eqn:approximated_lifted_space_dynamics}, we get
\begin{align}
   \underset{k\to\infty}{\lim}\dot{\widehat{f}}_k= \widehat{\boldsymbol{0}}.
\end{align}
Hence, the theorem is proved.
\end{proof}

\begin{theorem}
\normalfont For any $\boldsymbol{\omega}\in\mathcal{D}_{\widehat{\omega}}$ and $\boldsymbol{v}\in\mathcal{D}_{\widehat{v}}$, the following holds true
\begin{align}
    \|\widehat{f}_k\|>\|\widehat{f}_{k+1}\|, \quad k\in[2,N]_d
\end{align}
\end{theorem}
\begin{proof}
Since $\widehat{f}_k=\widehat{q}\tilde{\widehat{\omega}}^{B(k)}$, using Lemma 2, we have
\begin{align}
   \|\widehat{f}_{k+1}\|^2\leq {3 / 2}\|\widehat{q}\tilde{\widehat{\omega}}^{B(k)}\|^2\|\tilde{\widehat{\omega}}^{B}\|^2 ={3 / 2}\|\widehat{f}_k\|^2\|\tilde{\widehat{\omega}}^{B}\|^2
\end{align}
Therefore, we have
\begin{align}
    \|\tilde{\widehat{\omega}}^{B}\|^2=(\Bar{{0}},|\tilde{\bar{\boldsymbol{\omega}}}|^2)+(\Bar{{0}},|\tilde{\Bar{\boldsymbol{v}}}^2|)=(\Bar{\boldsymbol{0}},|\tilde{\bar{\boldsymbol{\omega}}}|^2+|\tilde{\Bar{\boldsymbol{v}}}^2|)<2/3  \nonumber
\end{align}
Hence,
\begin{align}
   \|\widehat{f}_{k+1}\|^2\leq {3 / 2}\|\widehat{q}\tilde{\widehat{\omega}}^{B(k)}\|^2\|\tilde{\widehat{\omega}}^{B}\|^2 <\|\widehat{f}_k\|^2.\nonumber
\end{align}
Subsequently, $\|f_{k+1}\|<\|f_k\|$. Hence the theorem follows.
\end{proof}
\section{Lifted linear state space model\label{sec:lifted_linear_state_space_model}}
Based on the derived observables in Section \ref{sec:observale_functions}, the lifted state space (from Theorem \ref{thm:observables}) is as follows:
\begin{align}
    \boldsymbol{z}=[\widehat{q},\;\;\widehat{\omega},\;\widehat{f}_1,\dots \widehat{f}_N]^\mathrm{T}.
\end{align}
The lifted state space $\boldsymbol{z}$ is used to learn the lifted state and input matrices, $A_\text{lift}$ and $B_\text{lift}$ which is described as follows. First, from a random uniform distribution $[-1,1]$ ,a set of random control inputs are chosen. These inputs are then applied sequentially to the discrete-time nonlinear system \eqref{eqn:discrete_nonlinear_equation} with $\boldsymbol{x}_0$ as the initial state to get the subsequent states. Let the control input $\widehat{\boldsymbol{u}}_k$ be applied to take the rigid body from $\boldsymbol{x}_k$ to $\boldsymbol{x}_{k+1}$. Consequently, we construct the matrices $\boldsymbol{X},\boldsymbol{U},$ and $\boldsymbol{Y}$ as follows:
\begin{align}
   &\boldsymbol{X}:=[\boldsymbol{x}_0,\dots,\boldsymbol{x}_{N_t-1}],\quad 
   \boldsymbol{U}:=[\widehat{\boldsymbol{u}}_0,\dots,\widehat{\boldsymbol{u}}_{N_t-1}],\nonumber\\ 
   &\boldsymbol{Y}:=[\boldsymbol{x}_1,\dots,\boldsymbol{x}_{N_t}].\nonumber
\end{align}
 where $N_{t}+1$ is the total number of data points collected. 
The matrix $\boldsymbol{Y}$ can be expressed as $  \boldsymbol{Y}$$=$$\boldsymbol{h}(\boldsymbol{X},\boldsymbol{U})$. Now given these matrices, $A_\text{lift}$ and $B_\text{lift}$, can be computed via the solution to the following optimization problem
\begin{align}
 \min _{A_\text{lift}, B_\text{lift}}\quad\left\|\boldsymbol{Y}_{\mathrm{lift}}-A_\text{lift} \boldsymbol{X}_{\mathrm{lift}}-B_\text{lift} \boldsymbol{U}\right\|_{F}  ,
 \label{eqn:least_squares_optimization}
\end{align}
where $\boldsymbol{X}_{\mathrm{lift}}=\left[\boldsymbol{z}(\boldsymbol{x}_{0}), \ldots, \boldsymbol{z}(\boldsymbol{x}_{N_t-1})\right]$ and $\boldsymbol{Y}_{\mathrm{lift}}=\left[\boldsymbol{z}(\boldsymbol{x}_{1}), \ldots, \boldsymbol{z}(\boldsymbol{x}_{N_t})\right]$.
The analytical solution to \eqref{eqn:least_squares_optimization} is given by $[A_\text{lift}, B_\text{lift}]=\boldsymbol{Y}_{\text {lift }}\left[\boldsymbol{X}_{\text {lift}}, \boldsymbol{U}\right]^{\dagger}$
where $(.)^{\dagger}$ denotes the Moore-Penrose pseudoinverse operator. Therefore, the lifted linear state space model is given as
\begin{align}
    \boldsymbol{z}_{k+1}=A_\text{lift}\boldsymbol{z}_k+B_\text{lift}\boldsymbol{u}_k
    \label{eqn:linear_lifted_state_space_model}
\end{align}
\begin{figure}[ht]
 \captionsetup[subfigure]{justification=centering}
 \centering
 \begin{subfigure}{0.23\textwidth}
{\includegraphics[scale=0.22]{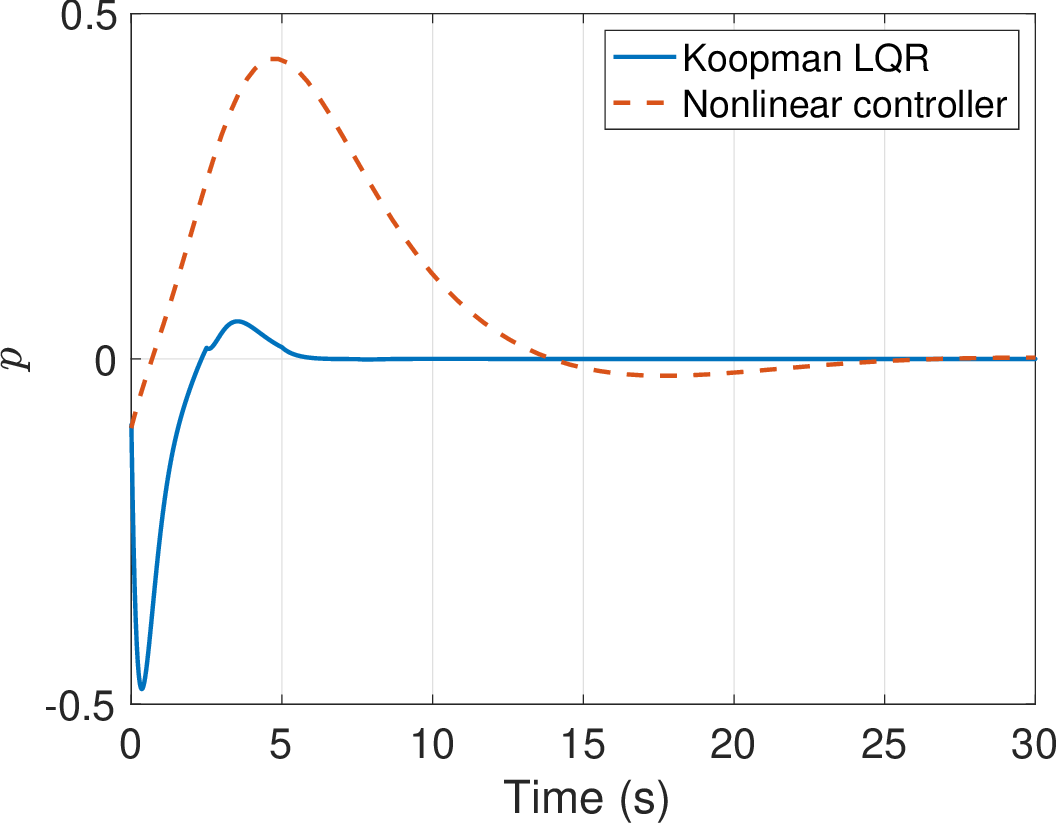}}
\caption{$\omega_1$}
\label{fig:}
 \end{subfigure}
\label{fig:}
 \begin{subfigure}{0.23\textwidth}
{\includegraphics[scale=0.22]{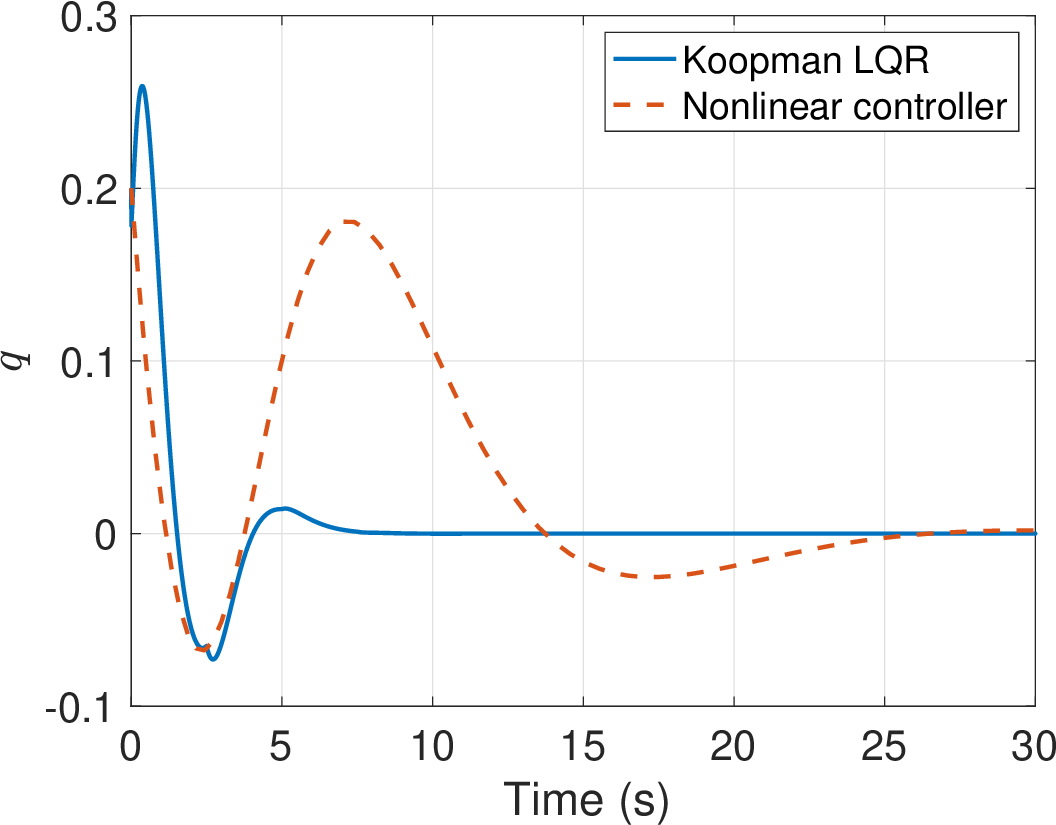}}
 \caption{$\omega_2$}
\label{fig:approximation_e_v}
 \end{subfigure}
\label{fig:}
 \begin{subfigure}{0.23\textwidth}
{\includegraphics[scale=0.22]{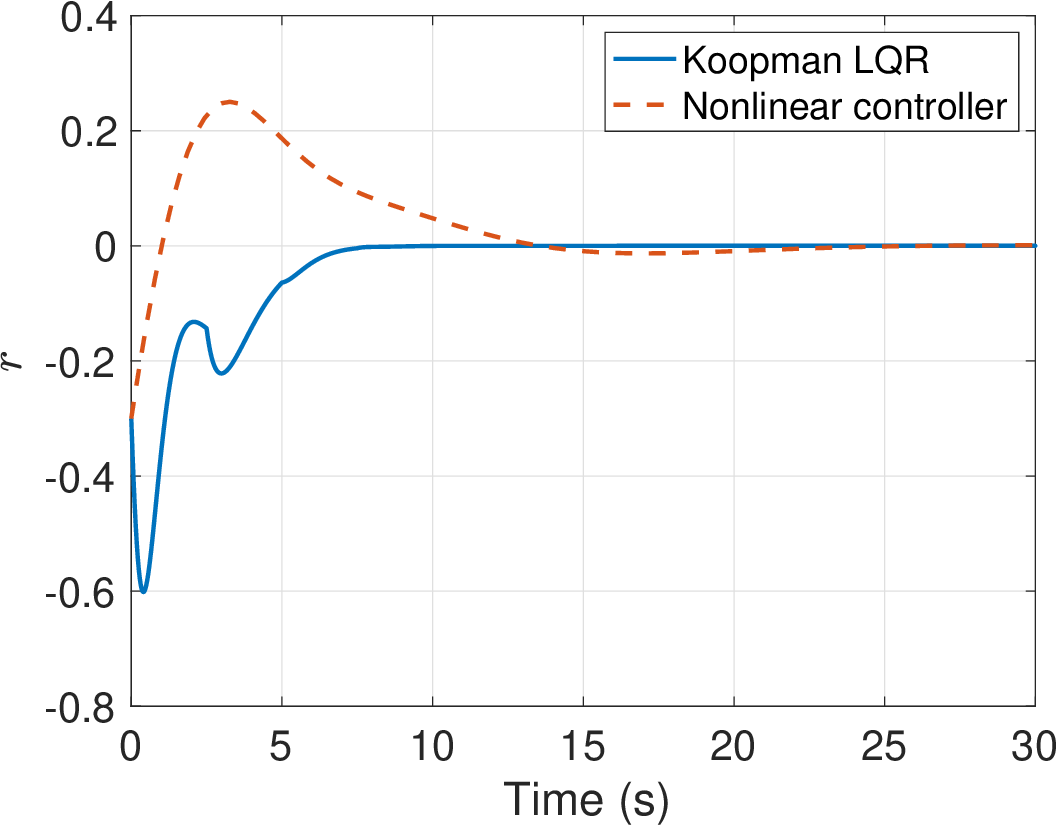}}
\caption{$\omega_3$}
\label{fig:}
\end{subfigure}
\begin{subfigure}{0.23\textwidth}
{\includegraphics[scale=0.22]{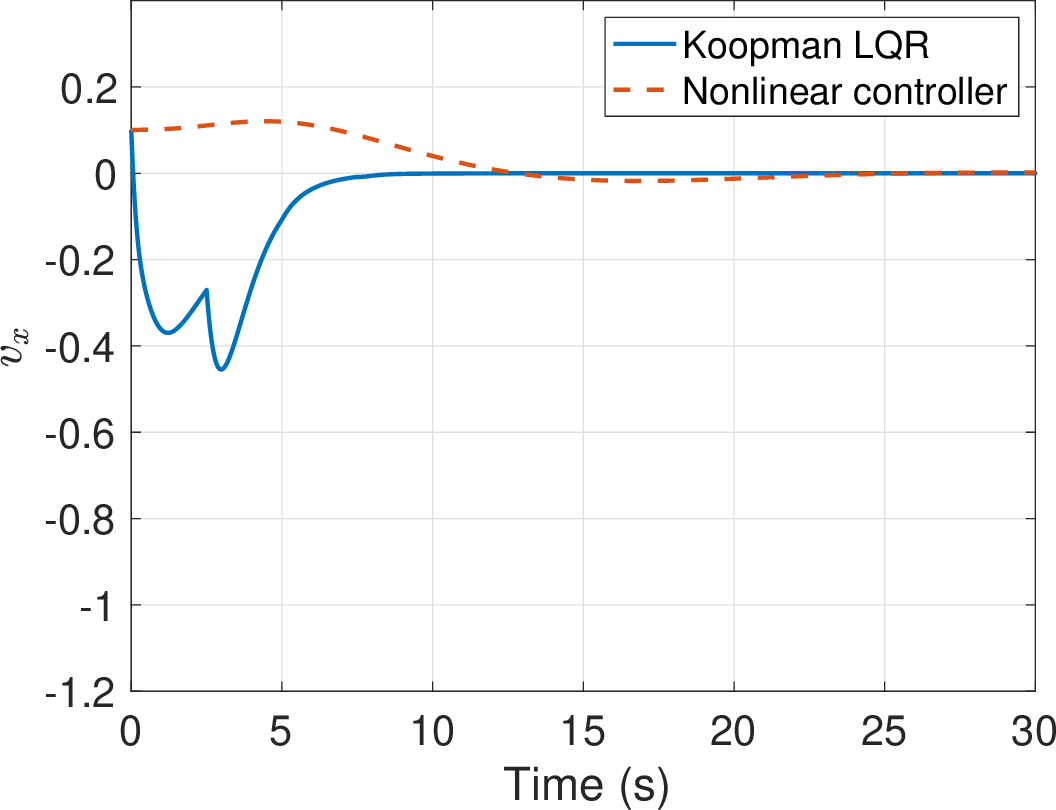}}
\caption{$v_1$}
\label{fig:}
 \end{subfigure}
\label{fig:}
 \begin{subfigure}{0.23\textwidth}
{\includegraphics[scale=0.22]{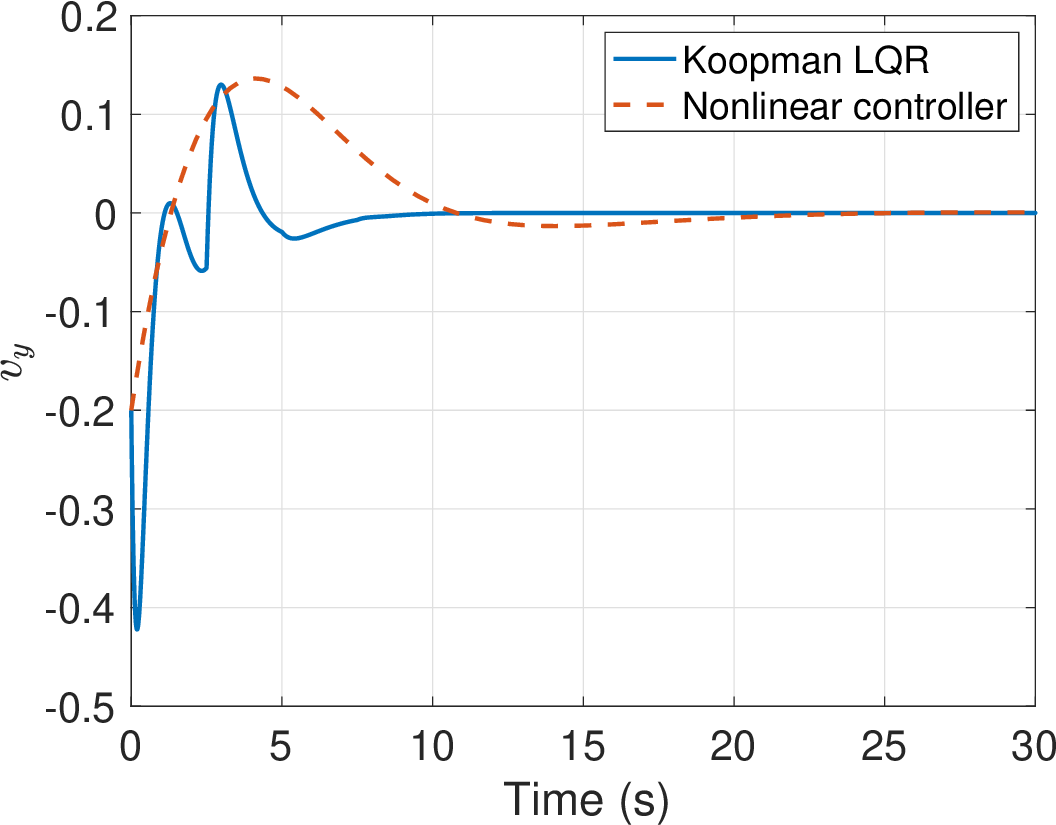}}
 \caption{$v_2$}
\label{fig:approximation_e_v}
 \end{subfigure}
\label{fig:}
 \begin{subfigure}{0.23\textwidth}
{\includegraphics[scale=0.22]{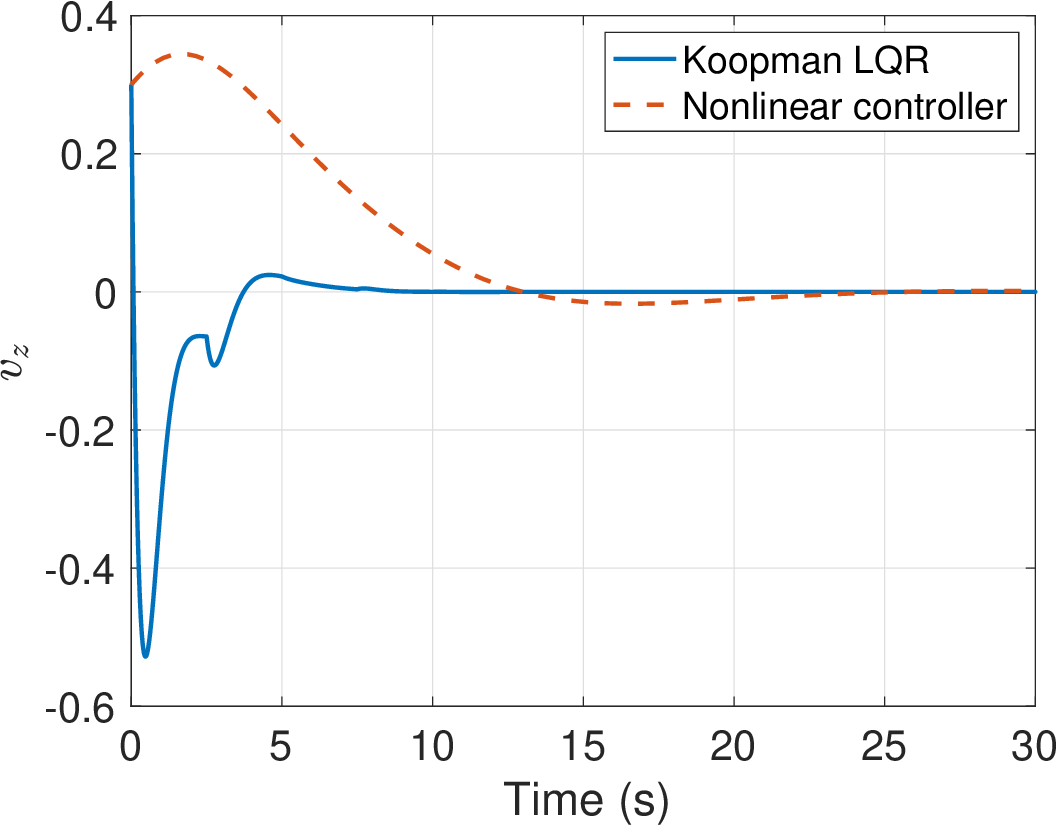}}
\caption{$v_3$}
\label{fig:}
\end{subfigure}
 \caption{Angular and linear velocities v/s time}
\label{fig:angular_velocity}
\end{figure}

\begin{figure}[]
 \captionsetup[subfigure]{justification=centering}
 \centering
 \begin{subfigure}{0.23\textwidth}
{\includegraphics[scale=0.22]{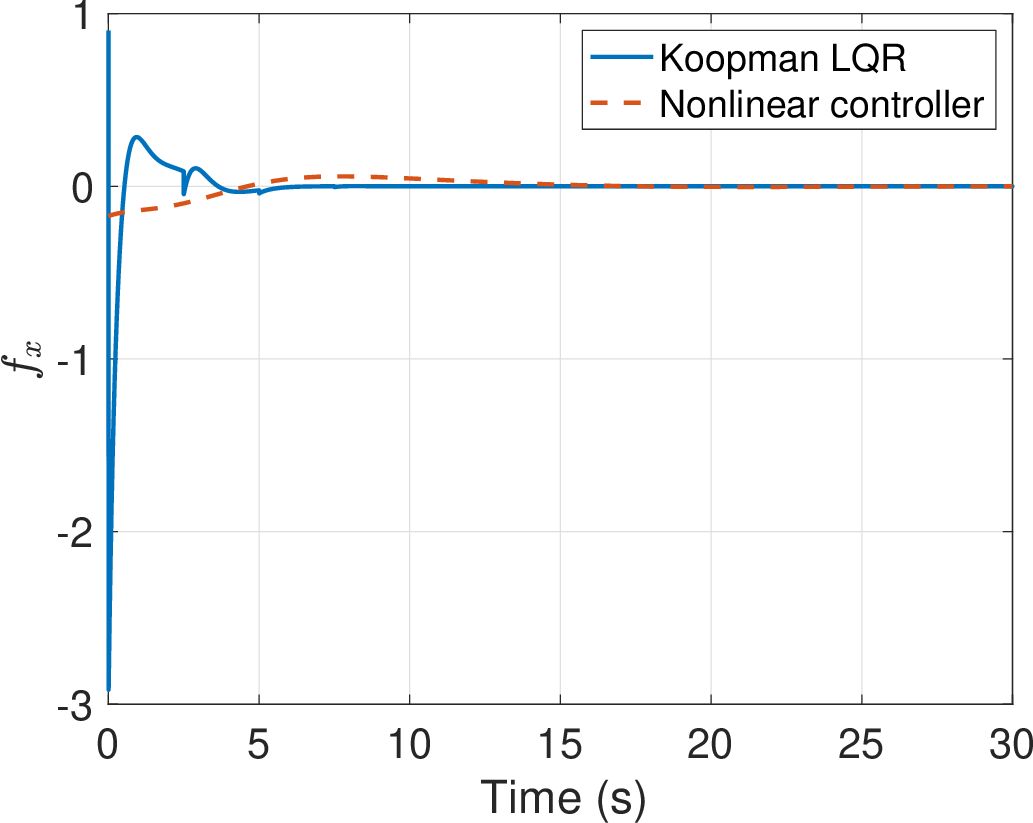}}
\caption{$f_x$ v/s time}
\label{fig:}
 \end{subfigure}
\label{fig:}
 \begin{subfigure}{0.23\textwidth}
{\includegraphics[scale=0.22]{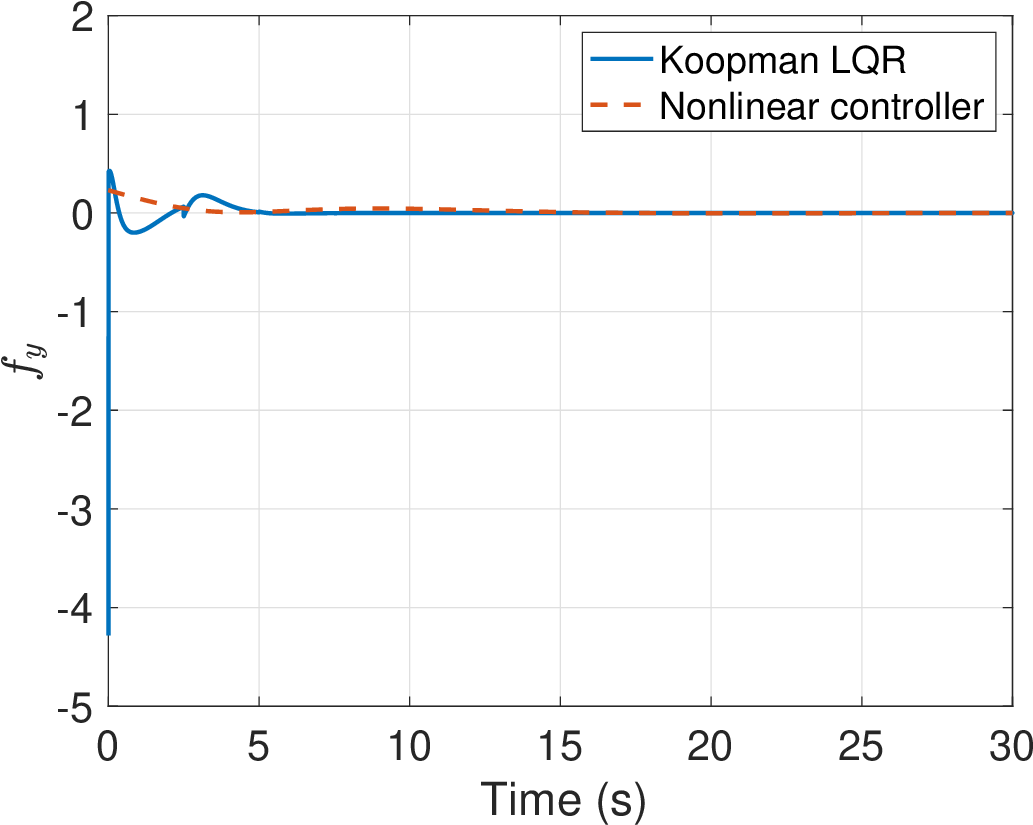}}
 \caption{$f_y$ v/s time}
\label{fig:approximation_e_v}
 \end{subfigure}
\label{fig:}
 \begin{subfigure}{0.23\textwidth}
{\includegraphics[scale=0.22]{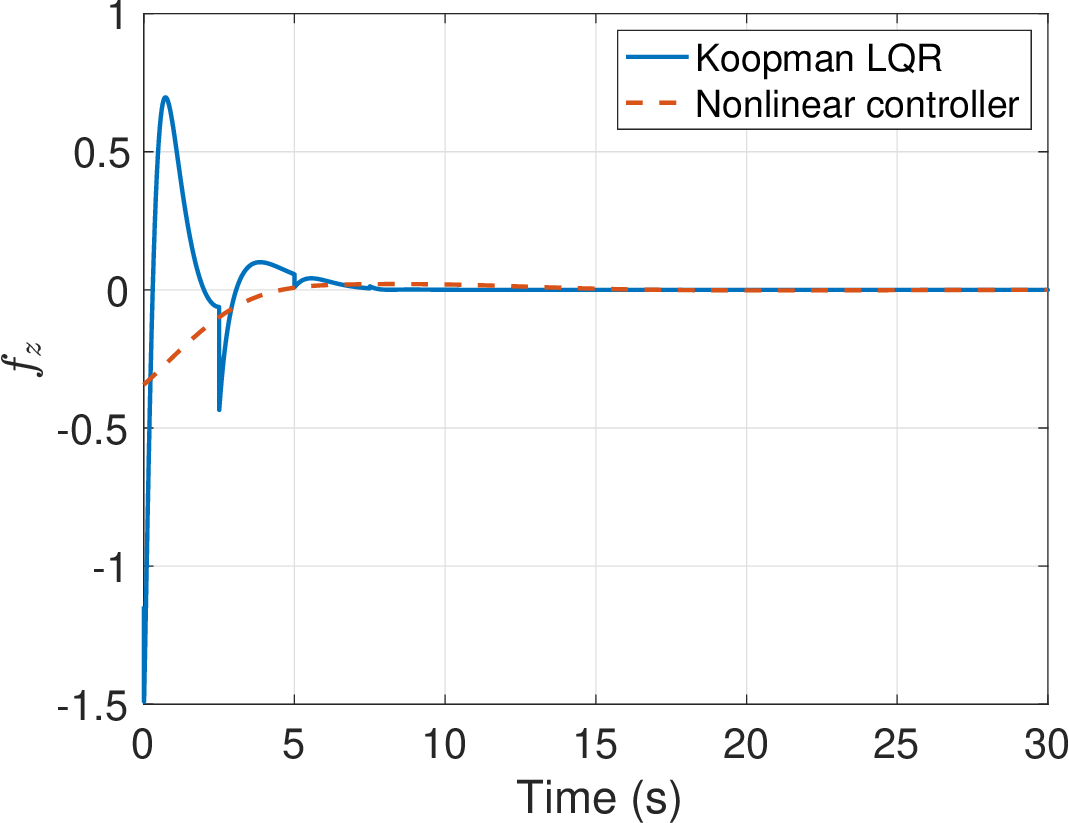}}
\caption{$f_z$ v/s time}
\label{fig:}
\end{subfigure}
\begin{subfigure}{0.23\textwidth}
{\includegraphics[scale=0.22]{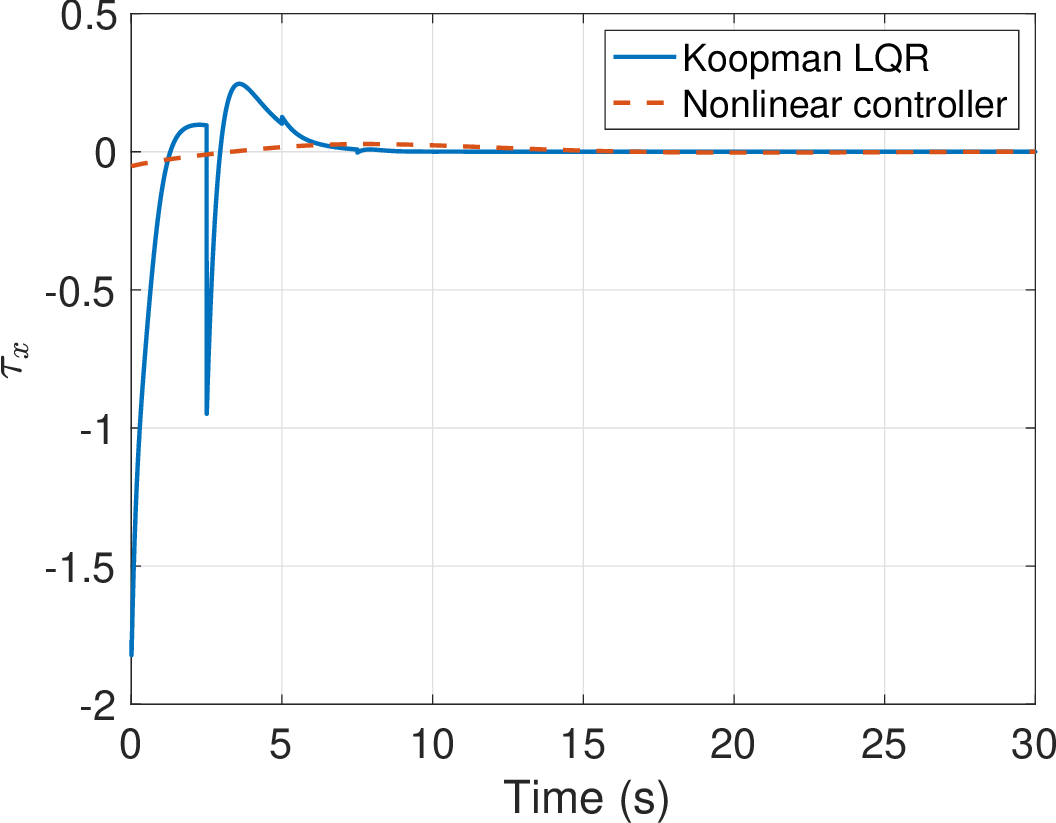}}
\caption{$\tau_x$ v/s time}
\label{fig:}
 \end{subfigure}
\label{fig:}
 \begin{subfigure}{0.23\textwidth}
{\includegraphics[scale=0.22]{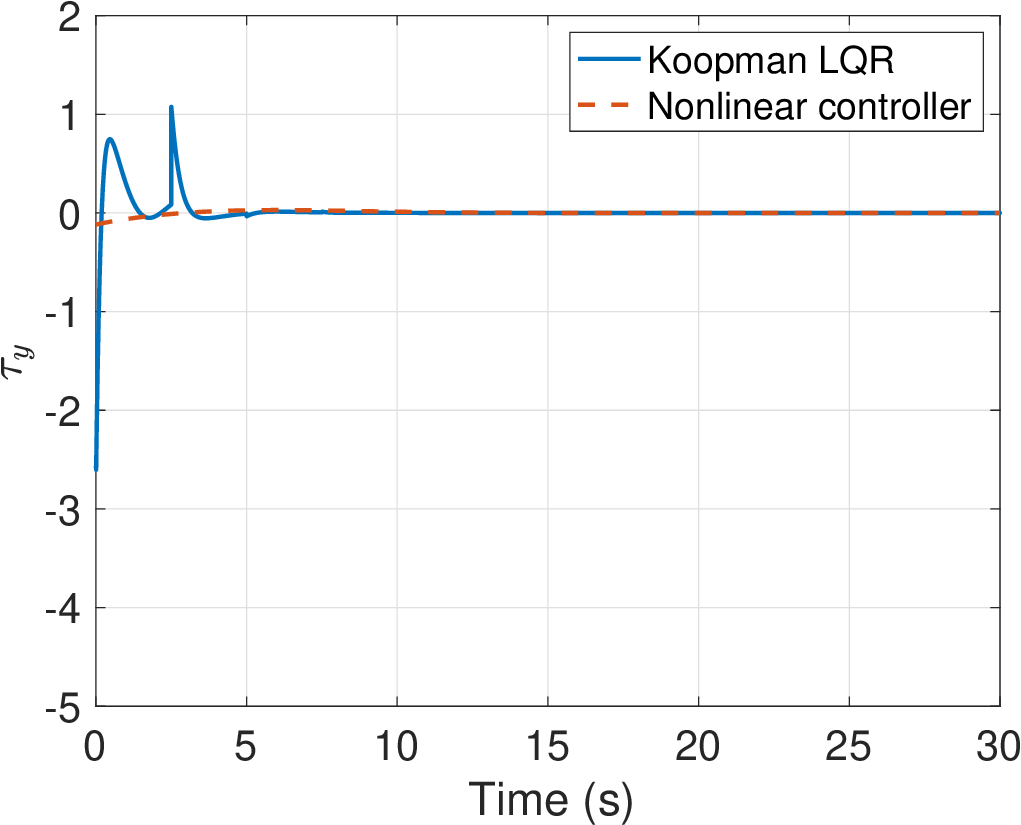}}
 \caption{$\tau_y$ v/s time}
\label{fig:approximation_e_v}
 \end{subfigure}
\label{fig:}
 \begin{subfigure}{0.23\textwidth}
{\includegraphics[scale=0.22]{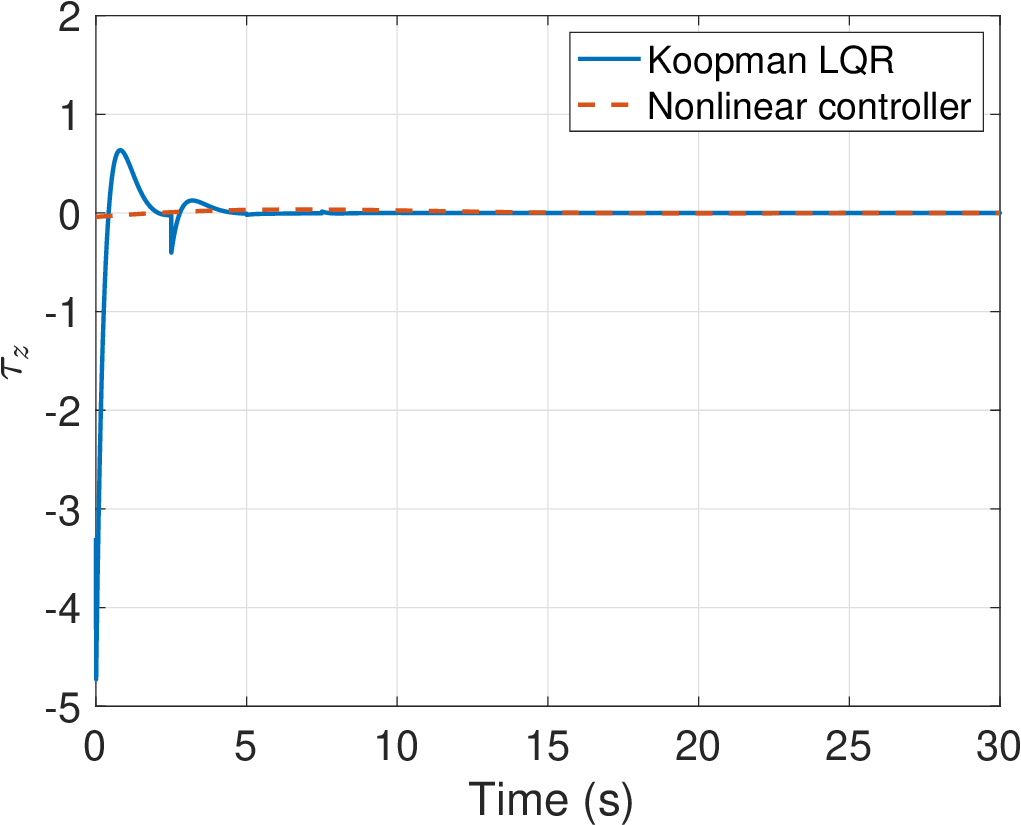}}
\caption{$\tau_z$ v/s time}
\label{fig:}
\end{subfigure}
 \caption{Force $F^B=[f_x\;\;f_y\;\;f_z]^\mathrm{T}$ and torque $\tau^B=[\tau_x\;\;\tau_y\;\;\tau_z]^\mathrm{T}$ versus time $t$}
\label{fig:}
\end{figure}

\section{Linear Control Design using LQR\label{sec:control_design}}
In this section, we design a LQR controller for the Koopman based on lifted state space model of the rigid body dynamics. Consider the lifted linear dynamics given by \eqref{eqn:linear_lifted_state_space_model}. The control design is based on the solution to the following infinite horizon LQR problem:
\begin{subequations}
\begin{align}
    &\min _{{\boldsymbol{z}}_k, \widehat{\boldsymbol{u}}_k}\quad\sum_{k=1}^{\infty}\;\; \boldsymbol{z}_k^\mathrm{T}Q_z\boldsymbol{z}_k+\boldsymbol{\widehat{u}}_k^\mathrm{T} R_z\widehat{\boldsymbol{u}}_k\\
    &\text{s.t.}\quad\boldsymbol{z}_{k+1}=A_\text{lift}\boldsymbol{z}_k+B_\text{lift}\widehat{\boldsymbol{u}}_k,
\end{align}
 \label{eqn:control_min_objective}
\end{subequations}
where $Q_z=Q_z^\mathrm{T}\succcurlyeq 0$ and $R_z=R_z^\mathrm{T}\succ 0$. During simulations, it was observed that the pairs $(A_{\text{lift}},Q_z^{\frac{1}{2}})$ and $(A_{\text{lift}},B_{\text{lift}})$ are observable and controllable respectively. The feedback control law that solves \eqref{eqn:control_min_objective} is given by
\begin{align}
    \widehat{\boldsymbol{u}}_k=-K\boldsymbol{z}_k,
\end{align}
where $K=\left(R_z+B_\text{lift}^{T} P B_\text{lift}\right)^{-1} B_\text{lift}^{T} P A_\text{lift}$ and $P$ satisfies the following discrete-time algebraic Riccati equation:
\begin{align}
\small
    P=&A_{\text{lift}}^{T} P A_{\text{lift}}-A_{\text{lift}}^{T} P B_{\text{lift}}\left(R_z+B^{T} P B_{\text{lift}}\right)^{-1} B_{\text{lift}}^{T} P A_{\text{lift}}\nonumber\\
    &+Q_z\nonumber
\end{align}

The control input $\boldsymbol{u}_k$ is then given by
\begin{align}
\boldsymbol{u}_k=M(\widehat{\boldsymbol{u}}_k+\widehat{\omega}\times(M\star((\dot{\widehat{\omega}})^s)))^s.
\end{align}
The control sequence $\{\boldsymbol{u}_k\}$ is then applied to the discrete-time nonlinear system \eqref{eqn:discrete_nonlinear_equation}. Algorithm \ref{alg:lqr_control} summarizes the control design for the dual quaternion based rigid body motion. The functions $\texttt{getKoopman}(\boldsymbol{z})$ takes in the initial state as $\boldsymbol{z}$ and returns $A_\text{lift}$ and $B_\text{lift}$ matrices as explained in Section \ref{sec:lifted_linear_state_space_model}. The function $\texttt{InfiniteLQR}(A_\text{lift},B_\text{lift},Q_z,R_z)$ computes the feedback control gain $K$ which minimizes \eqref{eqn:control_min_objective}. Lastly, the function $\texttt{rem}(a,b)$ returns the remainder when $a$ is divided by $b$.
\begin{algorithm}[H]
 \caption{Data-driven Koopman based LQR control}
 \small
\hspace*{\algorithmicindent} \textbf{Input:} $Q_z$, $R_z$, $\boldsymbol{x}_{0}$, $N_{\text{total}}$ and $N_t$ \\
 \hspace*{\algorithmicindent} \textbf{Output:} $\boldsymbol{z}$
\begin{algorithmic}[1]
 \State $\boldsymbol{z}_1=\boldsymbol{z}(\boldsymbol{x}_0)$
 \State $[A_\text{lift},B_\text{lift}]\gets \texttt{getKoopman}(\boldsymbol{z}_1)$
 \State $K\gets \texttt{InfiniteLQR}(A_\text{lift},B_\text{lift},Q_z,R_z)$
\For {$k=1\;\text{to}\;N_{\text{total}}$}
 \State $\widehat{\boldsymbol{u}}_k\gets -K\boldsymbol{z}$
 \State $\boldsymbol{u}_k\gets M(\widehat{\boldsymbol{u}}_k+\widehat{\omega}\times(M\star((\dot{\widehat{\omega}})^s)))^s$
 \State ${\boldsymbol{z}}_{k+1}\gets h(\boldsymbol{z}_k,{\boldsymbol{u}}_k)$
 \If {$\texttt{rem}(k,N_t)=0$}
 \State $[A_\text{lift},B_\text{lift}]\gets \texttt{getKoopman}(\boldsymbol{z}_{k+1})$
\State $K\gets \texttt{InfiniteLQR}(A_\text{lift},B_\text{lift},Q_z,R_z)$
 \EndIf
\EndFor
 \end{algorithmic}
\label{alg:lqr_control}
\end{algorithm}

A general control design structure is given in Fig. \ref{fig:feedback_diagram}.
\section{Numerical simulations\label{sec:numeical_simulations}}
Simulation studies have been carried out using MATLAB R2020b on an Intel Core i7 2.2GHz processor. The parameters for the rigid body are the same as in \cite{nuno_filipe2013rigid_nuno_rigid}. A rigid body with the moment of inertia
$\bar{I}^{B}=\left[\begin{smallmatrix}
1 & 0.1 & 0.15 \\
0.1 & 0.63 & 0.05 \\
0.15 & 0.05 & 0.85
\end{smallmatrix}\right] \quad \mathrm{Kg} \cdot \mathrm{m}^{2}$
and mass $m=1$kg is chosen. The rigid body is positioned at initial position $[x,y,z]^\mathrm{T}=[2,2,1]^\mathrm{T}$m with attitude $q=[q_1,q_2,q_3,q_4]^\mathrm{T}=[0.4618, 0.1917, 0.7999, 0.3320]^\mathrm{T}$. The initial linear and angular velocity in the body frame are equal to $\bar{v}^B=[v_x,v_y,v_z]^\mathrm{T}=[0.1, -0.2,0.3]^\mathrm{T}$m and $\bar{\omega}^B=[p,q,r]^\mathrm{T}=[-0.1,0.2,0.3]^\mathrm{T}$. The task is to steer the rigid body from the given initial state to the origin in the inertial frame. For LQR control design purposes, we take $Q_z=\texttt{blkdiag}(5\mathbf{I}_{16},\mathbf{0}_{N-16})$
and $R_z=\mathbf{I}_6$.
 The values of $N_t$ , $N_\text{total}$, and the sampling time $T$ are chosen as  $500s$, $6000s=30s/T$, and $0.05s$ respectively. Consequently, the feedback control inputs $\widehat{\boldsymbol{u}}$ computed from Algorithm \ref{alg:lqr_control} are added to the nonlinear system. Fig. \ref{fig:angular_velocity} shows the evolution of angular and linear velocities with time for $N_1=30s$
The lifted state space for the LQR based control design is chosen as follows:
\begin{align}
    \boldsymbol{z}=[\widehat{q}\;\;\widehat{q}\widehat{\omega}\;\;\widehat{q}\widehat{\omega}^{2}\;\;\widehat{q}\widehat{\omega}^{3}\;\;\widehat{q}\widehat{\omega}^{4}\;\;\widehat{q}\widehat{\omega}^{5}]^\mathrm{T}.\nonumber
\end{align}
\begin{table}
    \centering
\begin{tabular}{ | m{4em} | m{3cm}| m{2.6cm} | } 
  \hline
  LQR cost & Derived observables & Gaussian RBF's \\
  \hline
  $N=0$ & $7.9449\times 10^3$ & $7.9449\times 10^3$  \\ 
  \hline
  $N=3$ & $7.5697\times 10^3$  & $2.1485\times 10^5$\\ 
  \hline
  $N=5$ & $7.0287\times10^3$ & $3.2961\times 10^5$ \\
  \hline
\end{tabular}
\caption{LQR cost versus the number of observables $N$}
\label{table:cost}
\end{table}
As seen from Table \ref{table:cost}, the LQR cost decreases as the dimension of the lifted space increases. This is mainly because as $N$ increases, the lifted linear dynamics is able to better approximate the nonlinear dynamics. This is in agreement with our analysis. in addition, it is worth mentioning that, as seen from Table \ref{table:cost} using, other popular observables like the Gaussian radial basis functions (RBFs) might not always lead to decrease in the LQR cost as $N$ increases.

\section{Conclusions\label{sec:conclusions}}
In this paper, we derived a set of Koopman based observables for the rigid body dynamics based on the dual quaternion  representation which allowed us to describe the rigid body motion in terms of a linear state space model of higher dimension than the original nonlinear system (lifted state space model). Subsequently, we utilized the lifted linear model to design an LQR controller which turned out to perform in par with benchmark nonlinear controllers for stabilization of rigid body dynamics. The success of our approach can be justified theoretically by the fact that the proposed (infinite) set of Koopman based observables can form a sequence of functions that converges pointwise to the zero function. The latter statement practically means that a sufficiently large, yet finite, subset of the latter set of observables can be used to construct a linear state space model (approximation of the original nonlinear system) that describes the rigid body motion accurately enough for control design purposes. In our future work, we plan to utilize the proposed Koopman operator framework to design more sophisticated controllers (such as covariance steering algorithms) for uncertain rigid body dynamics.

\section{Appendix}
In this section, we present the discretization technique which is based on the fourth order Runga Kutta method. We discretize the continuous-time rigid body dynamics based on the dual quaternion representation \eqref{eqn:kinematics_dual} and \eqref{eqn:modified_control_input} using the classical fourth order Runga Kutta method as follows:
\begin{align}
\boldsymbol{x}_{k+1}&=\boldsymbol{x}_k+\frac{T}{6}\left(\boldsymbol{k}_{\mid 1}+2 \boldsymbol{k}_{\mid 2}+2 \boldsymbol{k}_{\mid 3}+\boldsymbol{k}_{\mid 4}\right),
\label{eqn:discrete_runga}
\end{align}
where $k\in[0, N-1]_d$, $T>0$ is the sampling period and $[a,b]_d=[a,b]\cap\mathbb{N}$. $\boldsymbol{k}_{\mid 1}$, $\boldsymbol{k}_{\mid 2}$, $\boldsymbol{k}_{\mid 3}$, and $\boldsymbol{k}_{\mid 4}$ are given as follows: 
\begin{subequations}
\begin{align}
& \boldsymbol{k}_{\mid 1}=f(\boldsymbol{x}_k,\boldsymbol{u}_k),\quad \boldsymbol{k}_{\mid 2}=f\left(\boldsymbol{x}_k+\frac{T}{2} \boldsymbol{k}_{\mid 1},\boldsymbol{u}_{k}\right),\nonumber\\ &\boldsymbol{k}_{\mid 3}=f\left(\boldsymbol{x}_k+\frac{T}{2} \boldsymbol{k}_{\mid 2},\boldsymbol{u}_k\right),\; \boldsymbol{k}_{\mid 4}=f\left(\boldsymbol{x}_k+T \boldsymbol{k}_{\mid 3},\boldsymbol{u}_k\right),\nonumber
\end{align}
\end{subequations}
 From \eqref{eqn:discrete_runga}, the discrete nonlinear rigid body dynamics can be written in compact form as follows
\begin{align}
    \boldsymbol{x}_{k+1}={h}(\boldsymbol{x}_k,\boldsymbol{u}_k).
    \label{eqn:discrete_nonlinear_equation}
\end{align}





 \bibliography{main.bib}

\end{document}